\documentclass[11pt,reqno,a4paper]{amsart}
\pdfoutput=1
\usepackage{graphicx}
\usepackage{amssymb}
\usepackage{amsthm}
\usepackage{amsmath}
\usepackage[left=1in,right=1in,top=1.2in,bottom=1.2in]{geometry}
\usepackage{mathtools}
\usepackage{caption}
\usepackage{hyperref}
\usepackage{xcolor}
\hypersetup{
    colorlinks,
    linkcolor={red!50!black},
    citecolor={blue!50!black},
    urlcolor={blue!80!black}
}
\usepackage{comment}
\usepackage{mdframed}
\definecolor{bgcolor}{rgb}{0.9,0.9,0.9}
\specialcomment{details}{\vspace{-0.5cm}\begin{mdframed}[backgroundcolor=bgcolor]}{\end{mdframed}}
\excludecomment{details}

\newtheorem{lemma}{Lemma}
\newtheorem{theorem}{Theorem}
\newtheorem{proposition}{Proposition}
\newtheorem{corollary}{Corollary}
\newtheorem{remark}{Remark}

\newcommand{\R}{\mathbb{R}}
\newcommand{\Z}{\mathbb{Z}}

\newcommand{\rmd}{\mathrm{d}}

\newcommand{\map}{\mathfrak{m}}
\newcommand{\Ei}{\mbox{Ei}}
\newcommand{\weight}{w}

\newcommand{\metricmaps}{\mathcal{X}}
\newcommand{\maps}{\mathcal{M}}
\newcommand{\edges}{\mathcal{E}}
\newcommand{\vertices}{\mathcal{V}}

\newcommand{\aut}{\mathrm{Aut}}
\newcommand{\exploration}{\mathfrak{e}}
\newcommand{\frontier}{\vec{\mathcal{F}}}
\newcommand{\cfun}{\mathcal{C}_g}
\newcommand{\cfunp}{\hat{\mathcal{C}}_g}
\newcommand{\rt}{\mathrm{r}}

\newcommand{\compactfrac}[2]{{\textstyle\frac{#1}{#2}}}
\newenvironment{compress}{\begingroup\medmuskip=0mu\relax}{\endgroup}

\begin{document}

\title{Multi-point functions of weighted cubic maps}

\author[J. Ambj\o rn and T.G. Budd]{J. Ambj\o rn$^{1,2}$ and T.G. Budd$^{1}$}
\dedicatory{
$^1$~The Niels Bohr Institute, Copenhagen University\\
Blegdamsvej 17, DK-2100 Copenhagen \O , Denmark.\\
\vspace{3mm}
$^2$~Institute for Mathematics, Astrophysics and Particle Physics (IMAPP)\\ 
Radbaud University Nijmegen, Heyendaalseweg 135,
6525 AJ, Nijmegen, The Netherlands \\
\vspace{3mm}
email: {\tt ambjorn@nbi.dk, budd@nbi.dk}
}

\begin{abstract}
\noindent
We study the geodesic two- and three-point functions of random weighted cubic maps, which are obtained by assigning random edge lengths to random cubic planar maps. 
Explicit expressions are obtained by taking limits of recently established bivariate multi-point functions of general planar maps. 
We give an alternative interpretation of the two-point function in terms of an Eden model exploration process on a random planar triangulation.
Finally, the scaling limits of the multi-point functions are studied, showing in particular that the two- and three-point functions of the Brownian map are recovered as the number of faces is taken to infinity.
\end{abstract}
\maketitle

\section{Introduction}

During the last few decades an increasingly intricate picture of the geometry of random surfaces has emerged. 
Of particular importance are the multi-point functions that describe the probability distributions of geodesic distances between random points in the surfaces.
The first such two-point function was derived in \cite{ambjom_scaling_1995} in the setting of random planar triangulations by a transfer-matrix approach. 
A striking result of this analysis is that in the limit of large triangulations the volume of a geodesic ball on average does not grow with the square of its radius, but with its fourth power, providing strong evidence that two-dimensional quantum gravity possesses a fractal dimension that is different from its topological dimension.

Later the discovery of a distance-preserving bijection between planar quadrangulations and certain labeled planar trees \cite{schaeffer_conjugaison_1998,chassaing_random_2004} sparked a renewed interest in the distance statistics of random surfaces.
Using this bijection the two-point function of random planar quadrangulations \cite{bouttier_geodesic_2003} was established rigorously.
It was extended to more families of planar maps with various restrictions on the degrees of the faces \cite{francesco_geodesic_2005,bouttier_planar_2012} and recently to \emph{bivariate} two-point functions of general planar maps \cite{ambjorn_trees_2013,bouttier_two-point_2013}, i.e. with prescribed number of vertices and faces.
Moreover, a clever application of an extension of the previously mentioned bijection \cite{miermont_tessellations_2009} allowed for the three-point function, measuring the pair-wise distances between three random vertices, to be established for random quadrangulations  \cite{bouttier_three-point_2008}. 
Very recently, using the techniques from \cite{ambjorn_trees_2013}, also this was extended to a bivariate three-point function for general planar maps \cite{fusy_three-point_2014}.

Each of these two- and three-point functions can be seen to have identical asymptotics (up to normalization of the geodesic distance) in the limit of large graphs.
This may be viewed as a consequence of the recently proven fact that various families of random maps, including planar maps with fixed face degrees \cite{gall_uniqueness_2013,miermont_brownian_2013} and general planar maps \cite{bettinelli_scaling_2013}, converge as metric spaces to a single random continuous metric space, known as the \emph{Brownian map}. 

In this paper we will study distances in yet another family of planar maps, namely cubic planar maps, i.e. planar maps for which the vertices all have degree three.
Currently an exact expression for the two-point function of such planar maps is unknown (although the transfer matrix approach in \cite{kawai_transfer_1993} comes close), likely due to the fact that the usual distance-preserving bijections are not easily adapted to a setting where the degree of the vertices is restricted.
However, the situation changes when one adapts the notion of distance in a planar map by introducing length variables on the edges of the planar maps, the result of which we call a \emph{weighted map}. 
When one takes the edge lengths to be independent exponential variables, the associated multi-point functions of weighted maps turn out to be related to a particular limit of the bivariate multi-point functions of general planar maps (as we will see in section \ref{sec:relation}).
This limit was first studied \cite{ambjorn_trees_2013} as a non-trivial scaling limit, known as \emph{generalized causal dynamical triangulations}, of planar quadrangulations with finitely many local maxima of the distance functions to a marked vertex.

It would be of general interest to find out whether random weighted cubic maps as random metric spaces converge to the Brownian map, when the number of vertices is taken to infinity, just like for the previously mentioned families of maps.
Showing this is beyond the scope of this investigation, but once we have the two- and three-point functions we can at least check that in the scaling limit they agree with those of the Brownian map, as we will do in section \ref{sec:scaling}.

The geometry of random cubic weighted maps is not only interesting in its own right, but some of its aspects have a direct interpretation in terms of well-known statistical systems that are generalized to live on a random planar map instead of a regular lattice.
\emph{First passage percolation} \cite{hammersley_first-passage_1965} describes distances on the (two-dimensional) regular lattice for which the lengths of the edges are taken to be independent (identically distributed) random variables.
In the case of exponentially distributed edge lengths this model is known to be closely related to the \emph{Eden model} \cite{eden_two-dimensional_1961}, which describes the random growth of a cluster of vertices. 
It turns out that, similarly, knowledge of the multi-point functions of weighted cubic maps provides detailed information about the statistics of an Eden model on a random planar triangulation (see section \ref{sec:eden}).

A particularly interesting property of the Eden model on the lattice is that the fluctuations in the shapes of the clusters are believed to be described by the Kardar--Parisi--Zhang universality class \cite{kardar_dynamic_1986} in a particular scaling limit.
An interesting question, which was raised in \cite{miller_quantum_2013}, is whether an analogous scaling limit exists when the Eden model is coupled to gravity, e.g. by putting it on a random triangulation or cubic planar map.
And if it exists, what are its scaling exponents?
In order to investigate these questions we need an understanding of the relation between the distance in a weighted cubic map to the graph distance of the corresponding unweighted map.
As a first step in this direction, we will show how a particular limit of the derived multi-point functions allow us derive bounds on the ratio of these distances.

Finally, we should mention the potential relevance of this work to the program set out  by Miller and Sheffield in \cite{miller_quantum_2013}.
They propose a \emph{Quantum Loewner Evolution} as the scaling limit of the Eden model on a random triangulation and hope to identify the explored regions of this process with the geodesic balls in the Brownian map.
Given the close connection between the Eden model on a random triangulation and the random weighed cubic maps, any evidence of the convergence of the latter to the Brownian map could be useful in achieving that goal.

The paper is organized as follows.
In Section \ref{sec:rwm} we introduce the notion of random planar maps with random edge weights and define the corresponding two- and three-point functions.
In Section \ref{sec:relation} we show that random weighted cubic maps (including any number of marked univalent vertices) can be exactly obtained from uniform random maps, in the limit of large vertex to face ratio, by an operation that involves the removal of dangling edges and merging of edges that share a bivalent vertex.
This precise relation is then used in Sections \ref{sec:twop} and \ref{sec:threepgen} to deduce the two- and three-point functions of weighted cubic maps with univalent marked vertices from the explicitly known (bivariate) two- and three point function for uniform planar maps. 
In subsequent Sections \ref{sec:coinc} and \ref{sec:coll} we study various limits of the three-point functions, which lead to alternative two-point functions, including one for which the marked vertices have higher degrees (Section \ref{sec:coinc}) and one where the geodesic has a marked vertex (Section \ref{sec:coll}).
In Section \ref{sec:eden} we study the relation between the Eden model exploration process on random triangulations and cubic weighted maps and give an alternative interpretation of the two-point function.
This interpretation allows us in particular to derive a differential equation relating the two-point functions with marked vertices of arbitrary degree.
A general solution to this differential equation is constructed in Section \ref{sec:gentwop}.
In section \ref{sec:scaling} we determine the scaling limits of the multi-point functions in the grand-canonical ensemble and compare them to the known expressions for the Brownian map.

\section{Random weighted maps}\label{sec:rwm}

For a planar map $\map$, denote its set of undirected edges by $\mathcal{E}(\map)$, its set of directed edges by $\vec{\mathcal{E}}(\map)$, and its set of vertices by $\mathcal{V}(\map)$.
A \emph{weighted map} $(\map,L)\in\metricmaps$ is a pair consisting of a planar map $\map$ and a length function $L:\mathcal{E}(\map) \to \R_+$.
One can associate to $(\map,L)$ a metric space $X_{\map,L}$ given by the quotient metric space of a disjoint union of intervals,
\begin{equation}
X_{\map,L} :=\Big(\bigsqcup_{e\in\mathcal{E}(\map)} [0,L(e)]\Big) / \sim,
\end{equation}
where $\sim$ appropriately identifies end-points of intervals according to the incidence relations of $\map$.
The metric space $X_{\map,L}$ comes with a natural measure $\rmd\mu$ originating from the Lebesgue measure on the intervals.

In this paper the edge lengths $L(e)$ will be independent random variables taken from an exponential distribution with unit expectation value, i.e. for a fixed planar map $\map$ the probability measure is
\begin{equation}\label{eq:lengthmeasure}
\prod_{e\in\edges(\map)} \rmd L(e)\, e^{-L(e)}\quad\text{on}\quad \R_+^{|\edges(\map)|}.
\end{equation} 
A measure on the space of weighted maps $\metricmaps$ is obtained by taking the product of this measure and a measure on the space of planar maps $\maps$.
In general the latter will be a restriction of the \emph{uniform measure} $\nu_{F,n}$ on the space  $\maps_{F,n}$ of planar maps with $F$ faces and $n$ (distinguished and distinct) marked vertices. 
By ``uniform'' we mean that a planar map $\map$ carries measure $\nu_{F,n}(\{\map\})=1/|\aut(\map)|$, where $\aut(\map)$ is the group of (orientation preserving) automorphisms of $\map$ preserving the marked vertices.
Equivalently, one may consider the corresponding measure on rooted planar maps, i.e. planar maps with a distinguished directed edge, with measure $1/(2|\edges(\map)|)$. 
However, the root will play no special role in the following and therefore we choose to stick to unrooted planar maps to reduce notational clutter.
In addition, we will consider (restrictions) of the \emph{grand canonical} measure $\nu_n$ on the space $\maps_n$ of planar maps with $n$ marked vertices, which includes a factor of $g$ for each face, i.e.
\begin{equation}
\nu_{n} = \sum_{F=1}^{\infty} g^F \nu_{F,n}.
\end{equation}

Several subclasses of planar maps will be of importance later.
In each case we will use the notation $\maps^\bullet_{F,n}$, with $\bullet$ replaced with some identifier, to denote a subset of planar maps equipped with the (appropriate restriction of the) measure $\nu_{F,n}(\map)$ and $\metricmaps^\bullet_{F,n}$ to denote the corresponding weighted maps with the product measure $\nu_{F,n}(\map,L)$ of $\nu_{F,n}(\map)$ and (\ref{eq:lengthmeasure}).
Similar notation will be used for the grand canonical measures $\nu^\bullet_n$ on $\maps_n^\bullet$ and $\metricmaps_n^\bullet$.

Let $\maps_{F,n}^{(3)}$ be the space of \emph{cubic} planar maps, where each vertex has degree $3$.
More generally, we define the space of \emph{almost cubic maps} $\maps_{F,n}^{(d_1,\cdots,d_n)}$ to contain the planar maps for which the $n$ marked vertices have degree $d_1,\cdots,d_n$, respectively, while all other are vertices are cubic.
The special case where all marked vertices have equal degree $d$ will be denoted by $\maps_{F,n}^{(d)}$.

\subsection{Multi-point functions}
Given any of these subclasses of planar maps, we can define the multi-point functions
\begin{equation}\label{eq:multipointgeneral}
G^\bullet_{F,n}((D_{ij})_{1\leq i<j\leq n}) := \int_{\metricmaps_{F,n}^\bullet} \rmd\nu^\bullet_{F,n}(\map,L,v_i) \prod_{1\leq i<j\leq n} \delta(D_{ij} - d_{X_{\map,L}}(v_i,v_j)).
\end{equation}
In particular, we will consider the \emph{two-point function}\footnote{The two- and three-point functions should of course be viewed as distributions on $\R$ and $\R^3$ respectively, but for convenience we will abuse notation and treat them as (generalized) functions. We could have chosen to consider instead the continuous \emph{cumulative two- and three-point functions}, for which the $\delta(\cdot)$ is replaced by a step function in (\ref{eq:twopointgeneral}) and (\ref{eq:threepointgeneral}), but the equations would become more cumbersome.}
\begin{equation}\label{eq:twopointgeneral}
G^\bullet_{F,2}(T) = \int_{\metricmaps_{F,2}^\bullet} \rmd\nu^\bullet_{F,2}(\map,L,v_1,v_2)\,\, \delta(T - d_{X_{\map,L}}(v_1,v_2)),
\end{equation}
and the \emph{three-point function}
\begin{align}
G^\bullet_{F,3}(D_{12},D_{23},D_{31}) = \int_{\metricmaps_{F,3}^\bullet} & \rmd\nu^\bullet_{F,3}(\map,L,v_1,v_2,v_3)\,\, \delta(D_{12} - d_{X_{\map,L}}(v_1,v_2))\nonumber\\
&\quad\times\delta(D_{23} - d_{X_{\map,L}}(v_2,v_3))\delta(D_{31} - d_{X_{\map,L}}(v_3,v_1)).\label{eq:threepointgeneral}
\end{align}
Notice that, by definition of the measures $\rmd\nu^\bullet_{F,n}$, the vertices $v_i$ in (\ref{eq:twopointgeneral}) and (\ref{eq:threepointgeneral}) are not allowed to coincide.
Similarly one may define the grand canonical multi-point functions $G^\bullet_{g,n}((D_{ij})_{1\leq i<j\leq n})$ by replacing the measure in (\ref{eq:multipointgeneral}) by $\nu_n^\bullet$ on $\metricmaps_n^\bullet$.

\begin{remark}
Another natural two-point function one can assign to weighted maps is the \emph{geometric} two-point function
\begin{equation}\label{eq:geomtwopoint}
G^{\mathrm{geom}}_{F,2}(T) := \int_{\metricmaps_{F,0}} \rmd\nu_{F,0}(\map,L) \int_{X_{\map,L}} \!\!\!\!\rmd\mu(x_1)\rmd\mu(x_2)\,\, \delta(T - d_{X_{\map,L}}(x_1,x_2)),
\end{equation}
where $\mu$ is the natural (Lebesgue) measure on $X_{\map,L}$. 
When the points $x_1$, $x_2$ sit in generic positions, one can naturally associate a weighted map in $\metricmaps_{F,n}^{(2)}$ to $(\map,L,x_1,x_2)$ by insertion of marked bivalent vertices at the points $x_1$, $x_2$.
As a consequence of the exponential distribution of the edge lengths, the measure $\rmd\nu_{F,0}(\map,L)\rmd\mu(x_1)\rmd\mu(x_2)$ is precisely mapped to the measure $\rmd\nu_{F,2}$ restricted to $\metricmaps_{F,n}^{(2)}$.
Therefore $G^{\mathrm{geom}}_{F,2}(T) = G^{(2)}_{F,2}(T)$ as long as $F\geq 3$.
Similarly, the natural generalization of (\ref{eq:geomtwopoint}) to the three-point function $G^{\mathrm{geom}}_{F,3}(D_{12},D_{23},D_{31})$ agrees with $G^{(2)}_{F,3}(D_{12},D_{23},D_{31})$.
Notice that, contrary to $G^{\mathrm{geom}}_{F,2}(T)$ and $G^{\mathrm{geom}}_{F,3}(D_{12},D_{23},D_{31})$, $G^{(2)}_{F,2}(T)$ and $G^{(2)}_{F,3}(D_{12},D_{23},D_{31})$ are non-zero for $F=2$, which will turn out to be more natural from the combinatorial point of view.
\end{remark}

The generating function for the number $W_F^{(3)}$ of rooted cubic planar maps with $F\geq 3$ faces can be found, among other places, in \cite{gao_number_1991} and is given by
\begin{equation}
\sum_{F=3}^{\infty} W_F^{(3)} g^F = \frac{1}{2} t^3(1-t)(1-4t+2 t^2), \quad \frac{1}{2}t(1-t)(1-2t)=g,
\end{equation}
where the root $t=\mathcal{O}(g)$ is chosen, and $W_F^{(3)}$ is given explicitly by
\begin{equation}\label{eq:numrooted}
W_F^{(3)} = 2^{2F-3} \frac{(3F-6)!!}{F!(F-2)!!}.
\end{equation}

\begin{lemma}
The measure $\nu_{F,n}(\metricmaps_{F,n}^{(d)})=\nu_{F,n}(\maps_{F,n}^{(d)})$ of the set of (unrooted) almost cubic planar maps with $F$ faces and $n$ marked vertices of degree $d$ is given for $d=3,2,1$ by
\begin{align}
\nu_{F,n}(\maps^{(3)}_{F,n}) &= 2^{2F-4} \frac{(2F-4)!(3F-8)!!}{F!(2F-4-n)!(F-2)!!}&(F\geq 3,\,0\leq n\leq 2F-4) \label{eq:measure3}\\
\nu_{F,n}(\maps^{(2)}_{F,n}) &= 2^{2F-4} \frac{(3F-7+n)!}{F!(3F-7)!!(F-2)!!} &(F\geq 2,\,n+F\geq 3) \label{eq:measure2}\\
\nu_{F,n}(\maps^{(1)}_{F,n}) &= 2^{2F-4+n} \frac{(3F-8+2n)!!}{F!(F-2)!!}& (F\geq 1,\,n+F\geq 3) \label{eq:measure1}
\end{align}
\end{lemma}
\begin{proof}
A cubic planar map with $F\geq 3$ faces has $|\edges(\map)|=3(F-2)$ edges and $|\vertices(\map)|=2(F-2)$ vertices.
Therefore $n$ (distinct) vertices can be marked in $(2F-4)!/(2F-4-n)!$ ways and (\ref{eq:measure3}) follows directly from (\ref{eq:numrooted}).

For $n\geq 1$ and $n+F\geq 4$, any almost cubic planar map in $\maps_{F,n}^{(2)}$ can be obtained uniquely from a planar map in $\maps_{F,n-1}^{(2)}$ by inserting a marked bivalent vertex in one of its $(3 F-7 +n)$ edges. 
Therefore $\nu_{F,n}(\maps^{(2)}_{F,n}) = (3F - 7+n) \nu_{F,n}(\maps^{(2)}_{F,n-1})$ and (\ref{eq:measure2}) follows from the fact that $\nu_{F,0}(\maps^{(2)}_{F,0})=\nu_{F,0}(\maps^{(3)}_{F,0})$ for $F\geq 3$, while the case $F=2$ can be checked by hand.
Similarly, for $n\geq 1$ and $n+F\geq 4$, any planar map in $\maps_{F,n}^{(1)}$ is uniquely obtained from a planar map in $\maps_{F,n-1}^{(2)}$ by inserting a cubic vertex in one of its $(3 F-8 +2n)$ edges and connecting it to a new marked 1-valent vertex.
Since the latter vertex can be on either side of the edge, we find that $\nu_{F,n}(\maps^{(1)}_{F,n}) = 2(3F - 8+2n) \nu_{F,n}(\maps^{(1)}_{F,n})$.
Again we deduce (\ref{eq:measure1}) from $\nu_{F,0}(\maps^{(1)}_{F,0})=\nu_{F,0}(\maps^{(3)}_{F,0})$ for $F\geq 3$, while the cases $F=1$ and $F=2$ are easily checked.
\end{proof}

\subsection{Properties of the multi-point functions}
Based on these results we can already determine a number of limits and integrals of the multi-point functions (\ref{eq:twopointgeneral}) and (\ref{eq:threepointgeneral}).
First of all, let us consider the $T\to 0$ limit of the two-point function $G_{F,2}^{\bullet}(T)$.
It is not hard to see that in this limit the only planar maps $(\map,L)\in\metricmaps^{\bullet}_{F,2}$ contributing to $G_{F,2}^{\bullet}(T)$ are those for which the marked vertices $v_1$ and $v_2$ are connected by an edge $e$ with length $L(e)\to 0$.
Hence we may write
\begin{equation}\label{eq:G2limit}
G^{\bullet}_{F,2}(0):=\lim_{T\to 0} G^{\bullet}_{F,2}(T) = \lim_{T\to 0}\int_{\metricmaps_{F,2}^\bullet} \rmd\nu_{F,2}(\map,L,v_1,v_2)\sum_{e\supset\{v_1,v_2\}} \delta(T - L(e)),
\end{equation}
where the sum is over edges $e$ connecting $v_1$ and $v_2$.

In the case $\map\in \metricmaps_{F,2}^{(1)}$ only the planar map $\map$ consisting of a single edge will contribute, leading to
\begin{equation}\label{eq:G2-1zero}
G_{F,2}^{(1)}(0) = \delta_{F,1}.
\end{equation}
On the other hand, if $d\geq 2$ and $\map\in \metricmaps_{F,2}^{(d)}$ contributes to (\ref{eq:G2limit}) with an edge $e\in \edges(\map)$, one can contract the edge to obtain a weighted map in $\metricmaps_{F,1}^{(2d-2)}$. 
Since any weighted map in $\metricmaps_{F,1}^{(2d-2)}$ can be obtained accordingly in exactly $2d-2$ ways, we find that
\begin{equation}\label{eq:G2-dzero}
G_{F,2}^{(d)}(0) = 2(d-1)\nu_{F,1}(\metricmaps^{(2d-2)}_{F,1}) \quad\quad(d\geq 2).
\end{equation}
More generally, for $d_1,d_2\geq 1$ we have
\begin{equation}\label{eq:G2-generalzero}
G_{F,2}^{(d_1,d_2)}(0) = (d_1+d_2-2)\nu_{F,1}(\metricmaps^{(d_1+d_2-2)}_{F,1}).
\end{equation}

One can check that the symmetry factors $1/|\aut(\map)|$ in the measure are appropriately taken care of in (\ref{eq:G2-dzero}) and (\ref{eq:G2-generalzero}) in the following way.
We could have replaced $1/|\aut(\map)|$ by $1/d_1$ and summed over planar maps that are rooted at an edge leaving the first vertex $v_1$.
Alternatively one can drop the factor $1/d_1$ and demand that the root edge is the shortest edge leaving $v_1$.
The symmetry factor in (\ref{eq:G2-generalzero}) is a result of the fact that contracting the shortest edge leaving $v_1$ results naturally in a planar map rooted at a (not necessarily shortest) edge leaving its marked vertex with degree $d_1+d_2-2$.
Very similar arguments will be used implicitly in the derivations to follow.
 
Finally, from their definition one can immediately deduce that the integrals of the two- and three-point functions are given by
\begin{align}
\int_0^{\infty}\rmd T\, G^{\bullet}_{F,2}(T) = \nu_{F,2}(\metricmaps^{\bullet}_{F,2} ), \label{eq:G2allT}\\
\int_0^{\infty}\rmd D_{12}\int_0^{\infty}\rmd D_{23}\int_0^{\infty}\rmd D_{31} G_{F,3}^{\bullet}(D_{12},D_{23},D_{31}) =\nu_{F,3}(\metricmaps^{\bullet}_{F,3} ).
\end{align}

\section{Weighted maps from general planar maps}\label{sec:relation}

Let $\map\in\maps_{F,n}$, $F+n\geq 3$, be a planar map with $F$ faces and $n$ (distinguished) marked vertices.
We define $\Phi(\map)$ to be the unique maximal submap of $\map$ containing the $n$ marked vertices while all its other vertices have degree larger than one, which can be obtained from $\map$ by repeatedly deleting \emph{dangling edges}, i.e. edges of which one of the extremities is unmarked and has degree one.
When $F+n\geq 3$, this determines a well-defined, idempotent map $\Phi:\maps_{F,n}\to\maps_{F,n}$ (see figure \ref{fig:phipsi}).

Given $\ell>0$, we can associate a weighted map $\Psi_\ell(\map):=(\map',L)$ to a planar map $\map\in\maps_{F,n}$ in the following way. 
We define the planar map $\map'$ to have vertex set $\vertices(\map')\subset\vertices(\map)$ given by all vertices of $\map$ except the unmarked bivalent vertices.
An edge $e\in\edges(\map')$ with length $L(e)=k \ell$ from $v_1$ to $v_2$ exists if and only if there is a chain of edges $\{e_1,\ldots,e_k\}\subset\edges(\map)$ of length $k$ from $v_1$ to $v_2$ such that for each $i=1,\ldots,k-1$, $e_i$ and $e_{i+1}$ share an unmarked bivalent vertex.
Since $\map'$ naturally inherits an embedding from $\map$, the weighted map $(\map',L)$ is well-defined and unique.
Therefore we have a map $\Psi_\ell : \maps_{F,n} \to \metricmaps_{F,n}$ (see figure \ref{fig:phipsi}).

The maps $\Phi(\map)$ and (the unweighted version of) $\Psi_\ell\circ\Phi (\map)$ are closely related to the so-called \emph{core} and \emph{kernel} of the planar map $\map$ (see e.g. \cite{noy_random_2014}), respectively, with the only difference that we require the marked vertices to be maintained. 

\begin{figure}[t]
\begin{center}
\includegraphics[width=\linewidth]{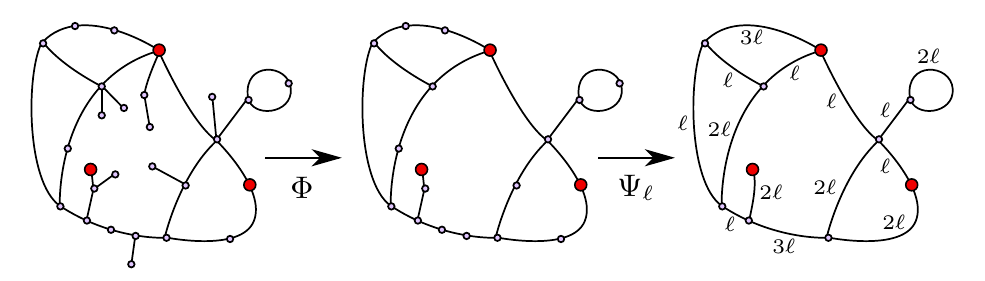}
\end{center}
\caption{The maps $\Phi:\maps_{F,n}\to\maps_{F,n}$ and $\Psi_\ell : \maps_{F,n} \to \metricmaps_{F,n}$ applied to a planar map with $F=6$ faces and $n=3$ marked vertices.}%
\label{fig:phipsi}
\end{figure}

Let $\nu_{F,n,x}$ be the discrete measure on the space $\maps_{F,n}$ of planar maps with $F$ faces and $n$ marked vertices such that each such planar map $\map$ carries measure $x^{|\edges(\map)|}/|\aut(\map)|$.

\begin{theorem}\label{thm:pushforward}
The pushforward measure $\ell(x)^{3F+2n-6}(\Psi_{\ell(x)}\circ\Phi)_*(\nu_{F,n,x})$ with $\ell(x) = \sqrt{4-16x}$ converges weakly as $x \to 1/4$ from below to the measure $\nu_{F,n}$ on the space $\metricmaps_{F,n}^{(1)}$ of weighted almost cubic maps with $n$ marked univalent vertices.
\end{theorem}  
\begin{proof}
We will first determine the discrete measure $\tilde{\nu}_{F,n,x}:=\Phi_*(\nu_{F,n,x})$. 
Let $\map\in\Phi(\maps_{F,n})$.
Each map in the preimage $\Phi^{-1}(\map)$ can be obtained uniquely from $\map$ by gluing a (potentially empty) rooted tree in each corner of $\map$. 
The number of rooted trees with $E$ edges is given by the Catalan number $\binom{2E}{E}/(E+1)$ with generating function $(1-\sqrt{1-4x}))/2x$.
Since the number of corners of $\map$ equals $2|\edges(\map)|$ we find that
\begin{equation}\label{eq:measurePhi}
\tilde{\nu}_{F,n,x}(\{\map\}) = \nu_{F,n,x}(\Phi^{-1}(\map)) = \frac{1}{|\aut(\map)|} \left(\frac{1-\sqrt{1-4x}}{2x}\right)^{2|\edges(\map)|} x^{|\edges(\map)|}.
\end{equation}
By construction the discrete measure $\tilde{\nu}_{F,n,x,\ell}:=\ell^{3F+2n-6}(\Psi_{\ell}\circ\Phi)_*(\nu_{F,n,x})$ has support on weighted maps $(\map,L)\in\metricmaps_{F,n}$ for which $\map$ contains no unmarked vertices with degree smaller than three and $L(e)/\ell\in\Z_+$ is a positive integer for each edge $e$.
Using (\ref{eq:measurePhi}) we find that for such a weighted map
\begin{equation}\label{eq:measurePsi}
\tilde{\nu}_{F,n,x,\ell}(\{(\map,L)\}) = \frac{\ell^{3F+2n-6}}{|\aut(\map)|} \left(\frac{(1-\sqrt{1-4x})^2}{4x}\right)^{\sum_{e\in\edges(\map)}L(e)/\ell}.
\end{equation}
Hence for fixed $\map$, the measure of the set of weighted maps $(\map,L)$ with arbitrary $L$ is given by
\begin{equation}\label{eq:marginalmeasure}
\frac{\ell^{3F+2n-6}}{|\aut(\map)|} \left(\frac{1}{\sqrt{4-16x}}-\frac{1}{2}\right)^{|\edges(\map)|}=\frac{\ell^{3F+2n-6}}{|\aut(\map)|} \left(\frac{1}{\ell(x)}-\frac{1}{2}\right)^{|\edges(\map)|}.
\end{equation}
Using Euler's formula $|\vertices(\map)|-|\edges(\map)|+F=2$, one finds that $|\edges(\map)| \leq 3F+2n-6$, while equality holds if and only if the $n$ marked vertices have degree one and all other vertices are cubic.
Therefore, if we set $\ell=\ell(x)$ and consider the limit as $x\to 1/4$ from below, we find in the latter case that (\ref{eq:marginalmeasure}) converges to $1/|\aut(\map)|$ and to zero otherwise.

It remains to show that for a fixed planar map $\map$ and a continuous function $f : \R_+^{|\edges(\map)|}\to \R$ we have
\begin{equation}\label{eq:measureconvergence}
\lim_{x\uparrow 1/4}\int_{\R_+^{|\edges(\map)|}} f(L) \tilde{\nu}_{F,n,x,\ell}(\map,\rmd L) = \int_{\R_+^{|\edges(\map)|}} f(L) \nu_{F,n}(\map,\rmd L).
\end{equation}
To leading order in $\ell=\ell(x)$, the right-hand side of (\ref{eq:measurePsi}) is given by
\begin{equation}
\frac{1}{|\aut(\map)|} \prod_{e\in\edges(\map)}\ell(x) e^{-L(e)}
\end{equation}
and therefore 
\begin{equation}
\lim_{x\uparrow 1/4}\int_{\R_+^{|\edges(\map)|}} f(L) \rmd\tilde{\nu}_{F,n,x,\ell} = \lim_{\ell\downarrow 0}\frac{1}{|\aut(\map)|} \sum_{\vec{k}\in \Z_+^{|\edges(\map)|}} \ell^{|\edges(\map)|} f(\ell\vec{k}) e^{-\ell \sum_{e\in\edges(\map)}\vec{k}_e},
\end{equation}
which indeed gives (\ref{eq:measureconvergence}).
\end{proof}

Recently, explicit generating functions have been found for the two-point functions \cite{ambjorn_trees_2013,bouttier_two-point_2013} and three-point functions \cite{fusy_three-point_2014} for general planar maps with fixed number of edges and faces.
In the following sections, we will show how Theorem \ref{thm:pushforward} can be used to derive the weighted map multi-point functions from these functions. 

\section{Two-point function}\label{sec:twop}

As in the previous section, let $\nu_{F,n,x}$ be the discrete measure $x^{|\edges(\map)|}/|\aut(\map)|$ on planar maps $\map\in\maps_{F,n}$ with $F$ faces and $n$ marked vertices.
The \emph{discrete two-point function} $\mathcal{G}_{F,x}:\Z \to \R$ for general planar maps is defined as
\begin{equation}
\mathcal{G}_{F,x}(t) := \int_{\maps_{F,2}} \rmd\nu_{F,2,x}(\map,v_1,v_2)\,\delta_{t,d_{\map}(v_1,v_2)},
\end{equation}
where $d_{\map}(v_1,v_2)$ is the graph distance between the marked vertices and $\delta_{i,j}=1$ iff $i=j$ and $\delta_{i,j}=0$ otherwise.
The corresponding generating function is denoted by
\begin{equation}
\mathcal{G}_{z,x}(t) := \sum_{F=1}^{\infty} z^F \mathcal{G}_{F,x}(t).
\end{equation}
In \cite{ambjorn_trees_2013,bouttier_two-point_2013,fusy_three-point_2014} an explicit expression for $\mathcal{G}_{z,x}(t)$ was found, which, borrowing some notation from \cite{fusy_three-point_2014}, reads
\begin{equation}\label{eq:discreteG2expr}
\mathcal{G}_{z,x}(t) = \log\left(\frac{[t+1]_{\sigma,a}^3[t+3]_{\sigma,a}}{[t]_{\sigma,a}[t+2]_{\sigma,a}^3}\right),\quad [t]_{\sigma,a} := 1 - a \,\sigma^t,
\end{equation}  
where $\sigma=\sigma(z,x)=x+\mathcal{O}(x^2)$ and $a=a(z,x)=z + \mathcal{O}(x)$ are the unique solutions that are analytic at the origin to
\begin{align}
x &= \frac{\sigma (1-a \sigma)^3(1-a \sigma^3)}{(1+\sigma+a \sigma-6a \sigma^2+a \sigma^3+a^2 \sigma^3+a^2 \sigma^4)^2},\label{eq:eqx}\\
z &= \frac{a(1-\sigma)^3(1-a^2 \sigma^3)}{(1-a \sigma)^3(1-a \sigma^3)}.\label{eq:eqz}
\end{align}

\begin{theorem}\label{thm:twopoint12}
The two-point functions $G_{g,2}^{(1)}(T)$ and $G_{g,2}^{(2)}(T)$, with $0<g\leq g^*:=1/(12\sqrt{3})$, for almost cubic weighted maps are given by
\begin{align}
G_{g,2}^{(1)}(T) &= \partial_T^3 \log\cfun(T), \label{eq:G2-1formula}\\
G_{g,2}^{(2)}(T) &= \compactfrac{1}{4}(1+\partial_T)^2\partial_T^3 \log\cfun(T),\label{eq:G2-2formula}
\end{align}
where
\begin{equation}\label{eq:Cgdef}
\cfun(T):= \Sigma \cosh \Sigma T+\alpha \sinh \Sigma T, \quad\quad \Sigma := \sqrt{\compactfrac{3}{2}\alpha^2-\compactfrac{1}{8}}
\end{equation}
and $\alpha$ is the largest positive solution to 
\begin{equation}\label{eq:eqalpha}
\alpha^3 - \alpha/4 + g = 0,
\end{equation}
or explicitly
\begin{equation}
\alpha = \frac{12 g}{1-\cos\left(\frac{2}{3}\arcsin(12\sqrt{3}g)\right)+\sqrt{3}\sin\left(\frac{2}{3}\arcsin(12\sqrt{3}g)\right)}.
\end{equation}
\end{theorem}
\begin{proof}
Let $\phi(T):\R_+\to \R$ be a smooth test function with support on some compact subinterval of $\R_+$.
Then by definition 
\begin{equation}\label{eq:smearG2}
\int_0^{\infty} \rmd T\,\phi(T) G_{F,2}^{(1)}(T) = \int_{\metricmaps_{F,2}^{(1)}} \rmd\nu_{F,2}(\map,L,v_1,v_2) \phi( d_{X_{\map,L}}(v_1,v_2) ).
\end{equation} 
Suppose $\ell>0$ and let $\map\in\maps_{F,2}$ be a planar map with marked vertices $v_1$ and $v_2$.
By construction, the distance $d_{X_{\Psi_\ell\circ\Phi(\map)}}(v_1,v_2)$ in the weighted map $\Psi_\ell\circ\Phi(\map)$ is given by $\ell d_{\map}(v_1,v_2)$.
Therefore, by Theorem \ref{thm:pushforward} one can express (\ref{eq:smearG2}) in terms of $\mathcal{G}_{F,x}(t)$ by
\begin{align}
\int_0^{\infty} \rmd T\,\phi(T) G_{F,2}^{(1)}(T) &= \lim_{x\uparrow 1/4} \ell(x)^{3F-2} \int_{\maps_{F,2}} \rmd\nu_{F,2,x}(\map,v_1,v_2) \phi(\ell(x)d_{\map}(v_1,v_2)) \nonumber\\
&= \lim_{\ell\to 0} \sum_{t=1}^{\infty} \ell^{3F-2} \phi(\ell t) \mathcal{G}_{F,x(\ell)}(t), \quad\quad x(\ell):=(1-\ell^2/4)/4.
\end{align}
For $g>0$ sufficiently small, one can turn this into a relation between the generating functions $G_{g,2}^{(1)}(T)$ and $\mathcal{G}_{z,x}(t)$,
\begin{align}
\int_0^{\infty} \rmd T\,\phi(T) G_{g,2}^{(1)}(T) &= \sum_{F=1}^{\infty} g^F\,\left[\lim_{\ell\to 0} \sum_{t=1}^{\infty} \ell^{3F-2} \phi(\ell t) \mathcal{G}_{F,(1-\ell^2/4)/4}(t)\right]\\
&=\lim_{\ell\to 0} \sum_{t=1}^{\infty} \ell^{-2} \phi(\ell t)\sum_{F=1}^{\infty}(\ell^3 g)^F \mathcal{G}_{F,(1-\ell^2/4)/4}(t) \\
&=\lim_{\ell\to 0} \sum_{t=1}^{\infty} \ell^{-2} \phi(\ell t)\, \mathcal{G}_{\ell^3 g,x(\ell)}(t),\label{eq:G2fromdiscrete}
\end{align}
where the second equality is justified by the fact that that there exists a $c>0$ and $B>0$ such that
\begin{equation}
\left|\sum_{t=1}^{\infty} \ell^{3F-2} \phi(\ell t) \mathcal{G}_{F,(1-\ell^2/4)/4}(t)\right|  < c\, \nu_{F,2}(\metricmaps_{F,2}^{(1)}) < c\, B^{-F}
\end{equation}
for all sufficiently small $\ell$.

According to (\ref{eq:G2fromdiscrete}) we should parametrize $x=x(\ell)$ and $z=\ell^3 g$ and consider the limit $\ell\to 0$ for fixed (sufficiently small) $g>0$.
Equations (\ref{eq:eqx}) and (\ref{eq:eqz}) are solved to leading order in $\ell$ by
\begin{equation}\label{eq:sigmaanda}
\sigma(\ell) = 1 - 2\Sigma \ell + \mathcal{O}(\ell^2),\quad a(\ell) = \frac{\alpha-\Sigma}{\alpha+\Sigma} + \mathcal{O}(\ell),
\end{equation}
where $\Sigma = \sqrt{\frac{3}{2}\alpha^2-\frac{1}{8}}$ and $\alpha = \alpha(g)$ is the largest positive solution to (\ref{eq:eqalpha}).
In particular, one finds that 
\begin{equation}\label{eq:limbracket}
\lim_{\ell\to 0}\,\log\left( [T/\ell]_{\sigma(\ell),a(\ell)}\right) = \log\left(1-\frac{\alpha-\Sigma}{\alpha+\Sigma}e^{-2\Sigma T}\right)
\end{equation} 
for $T>0$.

Using (\ref{eq:discreteG2expr}), we find that (\ref{eq:G2fromdiscrete}) is equal to
\begin{align}
\lim_{\ell\to 0}& \sum_{t=1}^{\infty} \ell^{-2}\left( 3\phi(\ell t-\ell) +\phi(\ell t-3\ell)-3\phi(\ell t-2\ell)-\phi(\ell t) \right)  \log[t]_{\sigma(\ell),a(\ell)} \\
&= - \lim_{\ell\to 0} \sum_{t=1}^{\infty} \ell \phi'''(\ell t) \log[t]_{\sigma(\ell),a(\ell)}\\
&= -\int_0^\infty \rmd T\, \phi'''(T)\, \log\left(1-\frac{\alpha-\Sigma}{\alpha+\Sigma}e^{-2\Sigma T}\right)\\
&= \int_0^\infty \rmd T\, \phi(T)\, \partial_T^3\log\cfun(T).
\end{align}
Since by construction $G_{g,2}^{(1)}(T)$ is regular as $T\to 0$, we conclude that $G_{g,2}^{(1)}(T)=\partial_T^3\log\cfun(T)$.

Let $(\map,L)\in\metricmaps_{F,2}^{(1)}$ be a weighted map with $F\geq 2$ faces and let $e_1$ and $e_2$ be the edges incident to the marked vertices.
Removing the edges $e_1$ and $e_2$ and marking the vertex or vertices at their other endpoint, one obtains a weighted map $(\map',L')$, which either (A) has two marked bivalent vertices, i.e. $(\map',L')\in \metricmaps_{F,2}^{(2)}$, or (B) one marked 1-valent vertex, i.e. $(\map',L')\in \metricmaps_{F,1}^{(1)}$ (see figure \ref{fig:twopointrelation}).
This observation leads to the decomposition $G_{g,2}^{(1)}(T) = g e^{-T} + G_{g,2}^{(A)}(T)+G_{g,2}^{(B)}(T)$, where the $g e^{-T}$ is the contribution of the weighted maps with $F=1$ faces.

\begin{figure}
\begin{center}
\includegraphics[width=.5\linewidth]{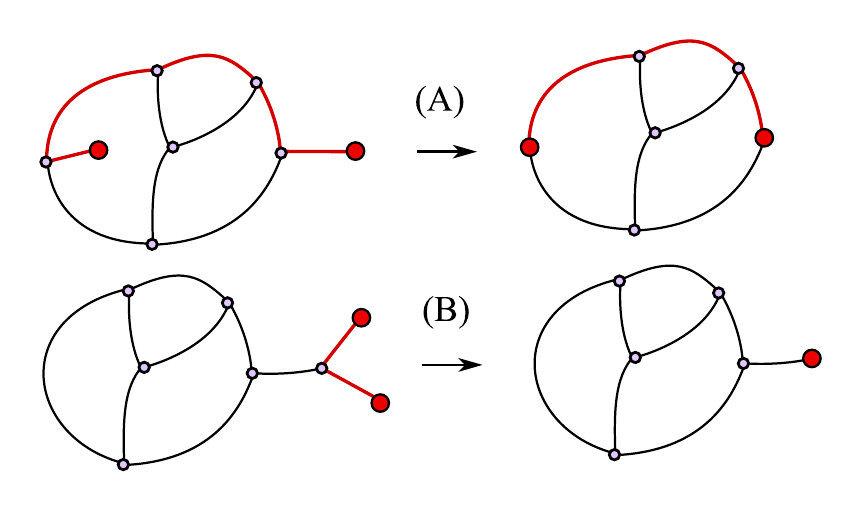}
\end{center}
\caption{Almost cubic weighted maps with two univalent marked vertices come in two types. Geodesics are shown in red.}%
\label{fig:twopointrelation}
\end{figure}

Notice that we can write
\begin{equation}\label{eq:twopointid}
G_{g,2}^{(A)}(T) = 4\int_0^{\infty}dL_0 e^{-L_0}\int_0^{\infty}dL_1 e^{-L_1} G_{g,2}^{(2)}(T-L_1-L_2),
\end{equation}
where the factor of $4$ comes from the fact that the edges $e_1$ and $e_2$ can be on either side of the marked vertices of $\map'$.
We can solve (\ref{eq:twopointid}) for $G_{g,2}^{(2)}(T)$, yielding
\begin{equation}
G_{g,2}^{(2)}(T) = \frac{1}{4}\left(1+\partial_T\right)^2G_{g,2}^{(A)}(T).
\end{equation}
On the other hand, it is clear that $G_{g,2}^{(B)}(T)$ will be of the form
\begin{equation}
G_{g,2}^{(B)}(T) \propto \int_0^{\infty}dL_0 e^{-L_0}\int_0^{\infty}dL_1 e^{-L_1} \delta(T-L_1-L_2) = T e^{-T}. 
\end{equation}
Therefore $(1+\partial_T)^2G_{g,2}^{(B)}(T)=0$ (and $(1+\partial_T)ge^{-T}=0$ of course) and (\ref{eq:G2-2formula}) follows.
\end{proof}
One can easily calculate
\begin{align}
\alpha &= \compactfrac{1}{2}-2 g-12 g^2-128 g^3-1680 g^4 + \mathcal{O}(g^5),\\
\Sigma &= \compactfrac{1}{2}-3 g-21 g^2-246 g^3-3453 g^4 + \mathcal{O}(g^5),
\end{align}
and therefore the first few terms of the two-point functions read
\begin{align}
G_{g,2}^{(1)}(T) &= g e^{-T}+g^2 \big(6 e^{-T} T+4 e^{-2 T}-4 e^{-T}\big)+g^3 \big(18 e^{-T}T^2+48 e^{-2 T} T +\nonumber\\
&\quad\quad +18 e^{-T} T+9 e^{-3 T}+40 e^{-2 T}-49e^{-T}\big)+\mathcal{O}(g^4),\\
G_{g,2}^{(2)}(T) &= g^2 e^{-2 T}+g^3 \left(12 e^{-2 T} T+9 e^{-3 T}-14 e^{-2 T}+9e^{-T}\right)+\mathcal{O}(g^4).
\end{align} 

Let us check some properties of $G_{g,2}^{(1)}(T)$ and $G_{g,2}^{(2)}(T)$.
First of all, their integrals are given by
\begin{align}
\int_0^{\infty}\rmd T\,G_{g,2}^{(1)}(T)&=-\left.\partial_T^2\log\cfun(T)\right|_{T=0}=\alpha^2-\Sigma^2 = \frac{g}{2\alpha}\nonumber\\
&=\frac{1}{24}\left(1-\cos\left(\textstyle \frac{2}{3}\arcsin(12\sqrt{3}g)\right)+\sqrt{3}\sin\left(\textstyle\frac{2}{3}\arcsin(12\sqrt{3}g)\right)\right),\label{eq:G2-1intgr}\\
\int_0^{\infty}\rmd T\,G_{g,2}^{(2)}(T)&=-\frac{1}{4} \left.(1+\partial_T)^2\partial_T^2\log\cfun(T)\right|_{T=0} = \frac{g (1-2\alpha)(5-6\alpha)}{32\alpha},\label{eq:G2-2intgr}
\end{align}
which can be checked to be generating functions of (\ref{eq:measure1}) and (\ref{eq:measure2}) respectively for $n=2$, confirming (\ref{eq:G2allT}).

On the other hand,
\begin{align}
G_{g,2}^{(1)}(0)&=\left.\partial_T^3\log\cfun(T)\right|_{T=0}=2\alpha(\alpha^2-\Sigma^2) = g, \label{eq:G2-1zero2}\\
G_{g,2}^{(2)}(0)&=\frac{1}{4} \left.(1+\partial_T)^2\partial_T^3\log\cfun(T)\right|_{T=0} = \frac{g (1-2\alpha)(6\alpha-1)}{16\alpha}.\label{eq:G2-2zero}
\end{align}
Clearly the former is in agreement with (\ref{eq:G2-1zero}). 
To see that (\ref{eq:G2-2zero}) agrees with (\ref{eq:G2-dzero}), notice that according to (\ref{eq:measure2}) we have $\nu_n(\metricmaps_n^{(2)}) = (3 g\frac{\partial}{\partial g}-5) \nu_n(\metricmaps_n^{(1)})$.
Using $\frac{\partial\alpha}{\partial g} = 4/(1-12\alpha^2)$, which follows from (\ref{eq:eqalpha}), we can easily check that 
\begin{equation}
\int_0^{\infty}\rmd T\,G_{g,2}^{(2)}(T) = \frac{1}{2}\left(3 g\frac{\partial}{\partial g}-5\right) G_{g,2}^{(2)}(0)
\end{equation}
is satisfied by the expressions (\ref{eq:G2-2intgr}) and (\ref{eq:G2-2zero}).

In section \ref{sec:eden} the expression for $G_{g,2}^{(2)}$ will be rederived using a different perspective, which will clarify the appearance of the differential operator $(1+\partial_T)$ and will allow us to derive an expression for $G_{g,2}^{(3)}(T)$.

\section{Three-point function}

\subsection{General expression}\label{sec:threepgen}
The steps in the previous section can be repeated to obtain the three-point function from its discrete counter part $\mathcal{G}_{z,x}(d_{12},d_{23},d_{31})$ for general planar maps.
It will be convenient to use the parametrization
\begin{equation}
D_{12}=S+T,\quad D_{23}=T+U,\quad D_{31}=U+S
\end{equation}
and to define the corresponding three-point functions
\begin{equation}
\bar{G}_{g,3}^{\bullet}(S,T,U) := 2G_{g,3}^{\bullet}(S+T,T+U,U+S).
\end{equation}

In \cite{fusy_three-point_2014} an explicit expression was found for 
\begin{equation}
\mathcal{G}_{z,x}(d_{12},d_{23},d_{31}) := \sum_{F=1}^{\infty}z^F \!\!\int_{\maps_{F,3}}\!\!\!\rmd\nu_{F,3,x}(\map,v_1,v_2,v_3)\delta_{d_{12},d_{\map}(v_1,v_2)}\delta_{d_{23},d_{\map}(v_2,v_3)}\delta_{d_{31},d_{\map}(v_3,v_1)}.
\end{equation}
It can be summarized as\footnote{The $1+\cdot$ is there to compensate for the fact that the irrelevant constant term present in $\mathcal{F}^{\mathrm{even}}_{z,x}(s,t,u)$ in \cite{fusy_three-point_2014} was chosen such that $\mathcal{F}^{\mathrm{even}}_{z,x}(0,0,0)=1$.}
\begin{align}
1+\sum_{s'=0}^s\sum_{t'=0}^t\sum_{u'=0}^u \mathcal{G}_{z,x}(s'+t',t'+u',u'+s') &= \mathcal{F}^{\mathrm{even}}_{z,x}(s,t,u),\label{eq:sumeven}\\
\sum_{s'=1}^s\sum_{t'=1}^t\sum_{u'=1}^u \mathcal{G}_{z,x}(s'+t'-1,t'+u'-1,u'+s'-1) &= \mathcal{F}^{\mathrm{odd}}_{z,x}(s,t,u),\label{eq:sumodd}
\end{align}
where 
\begin{compress}
\begin{align}
\mathcal{F}^{\mathrm{even}}_{z,x}(s,t,u)&:=\frac{[3]_{\sigma,a}([s+2]_{\sigma,a}[t+2]_{\sigma,a}[u+2]_{\sigma,a}[s+t+u+3]_{\sigma,a})^2}{[2]_{\sigma,a}^3[s+t+2]_{\sigma,a}[t+u+2]_{\sigma,a}[u+s+2]_{\sigma,a}[s+t+3]_{\sigma,a}[t+u+3]_{\sigma,a}[u+s+3]_{\sigma,a}},\\
\mathcal{F}^{\mathrm{odd}}_{z,x}(s,t,u)&:=\frac{\sigma^3[3]_{\sigma,a}(a[s]_{\sigma,1}[t]_{\sigma,1}[u]_{\sigma,1}[s+t+u+3]_{\sigma,a^2})^2}{[2]_{\sigma,a}^3[s+t+2]_{\sigma,a}[t+u+2]_{\sigma,a}[u+s+2]_{\sigma,a}[s+t+3]_{\sigma,a}[t+u+3]_{\sigma,a}[u+s+3]_{\sigma,a}},
\end{align}
\end{compress}
using the same notation $[t]_{\sigma,a}:=1-a\sigma^t$ as before.

\begin{theorem}\label{thm:threepoint}
The three-point functions $\bar{G}_{g,3}^{(1)}(S,T,U)$ and $\bar{G}_{g,3}^{(2)}(S,T,U)$ for weighted cubic maps are given by
\begin{align}
\bar{G}_{g,3}^{(1)}(S,T,U) =&\, \partial_S\partial_T\partial_U F_g(S,T,U)\\
\bar{G}_{g,3}^{(2)}(S,T,U) =&\,\compactfrac{1}{8}(1+\partial_S)(1+\partial_T)(1+\partial_U)\bar{G}_{g,3}^{(1)}(S,T,U) \nonumber\\
&+ G_{g,2}^{(2)}(T+U)\delta(S) + G_{g,2}^{(2)}(U+S)\delta(T) + G_{g,2}^{(2)}(S+T)\delta(U),\label{eq:G3-2expr}
\end{align}
where $G_{g,2}^{(2)}(T)$ is two-point function from Theorem \ref{thm:twopoint12} and 
\begin{align}
F_g(S,T,U) &:= F^{\mathrm{even}}_g(S,T,U) + F^{\mathrm{odd}}_g(S,T,U),\\
F^{\mathrm{even}}_g(S,T,U)&:= \frac{1}{\Sigma^2}\,\frac{\cfun^2(S)\cfun^2(T)\cfun^2(U)\cfun^2(S+T+U)}{\cfun^2(S+T)\cfun^2(T+U)\cfun^2(U+S)},\\
F^{\mathrm{odd}}_g(S,T,U)&:=\frac{(\alpha^2-\Sigma^2)^2}{\Sigma^2}\,\,  \frac{\sinh^2\Sigma S\sinh^2\Sigma T\sinh^2\Sigma U\,\cfunp(S+T+U)^2}{\cfun^2(S+T)\cfun^2(T+U)\cfun^2(U+S)},\\
\cfunp(T)&:=2\alpha \Sigma \cosh(\Sigma T)+(\alpha^2+\Sigma^2)\sinh(\Sigma T)
\end{align}
and $\cfun(T)$ as in (\ref{eq:Cgdef}).
\end{theorem}

\begin{proof}
Setting $x=x(\ell)=(1-\ell^2/4)/4$, $z=\ell^3 g$, $s=S/\ell$, $t=T/\ell$ and $u=U/\ell$, one can check that we have the following limits for $S,T,U,g$ fixed:
\begin{align}
\lim_{\ell\to 0}\mathcal{F}^{\mathrm{even}}_{\ell^3g,x(\ell)}\left(\lfloor S/\ell\rfloor,\lfloor T/\ell\rfloor,\lfloor U/\ell\rfloor\right) &= F^{\mathrm{even}}_g(S,T,U),\label{eq:evenlim}\\
\lim_{\ell\to 0}\mathcal{F}^{\mathrm{odd}}_{\ell^3g,x(\ell)}\left(\lfloor S/\ell\rfloor,\lfloor T/\ell\rfloor,\lfloor U/\ell\rfloor\right) &= F^{\mathrm{odd}}_g(S,T,U),\label{eq:oddlim}.
\end{align}
By construction we have the identity
\begin{compress}
\begin{align}
&\int_0^T\rmd T'\int_0^S\rmd S'\int_0^U\rmd U' \bar{G}_{g,3}^{(1)}(S',T',U') = \int_0^{\infty}\!\!\rmd D_{12}\int_0^{\infty}\!\!\rmd D_{23}\int_0^{\infty}\!\!\rmd D_{31} G_{g,3}^{(1)}(D_{12},D_{23},D_{31}) \nonumber\\
&\quad\times \theta(S - (D_{12}+D_{31}-D_{23})/2)\,\theta(T - (D_{23}+D_{12}-D_{31})/2)\,\theta(U - (D_{31}+D_{23}-D_{12})/2).
\end{align}
\end{compress}
Application of Theorem \ref{thm:pushforward} shows that the right-hand side can be expressed in terms of the discrete three-point function $\mathcal{G}_{z,x}(d_{12},d_{23},d_{31})$ as
\begin{compress}
\begin{align}
&\int_0^T\rmd T'\int_0^S\rmd S'\int_0^U\rmd U' \bar{G}_{g,3}^{(1)}(S',T',U') = \lim_{\ell\to 0} \sum_{d_{12}=1}^{\infty} \sum_{d_{23}=1}^{\infty}\sum_{d_{31}=1}^{\infty}\mathcal{G}_{\ell^3g,x(\ell)}(d_{12},d_{23},d_{31})\nonumber\\
&\quad\times \theta(S' - \ell(d_{12}+d_{31}-d_{23})/2)\,\theta(T' - \ell(d_{23}+d_{12}-d_{31})/2)\, \theta(U' - \ell(d_{31}+d_{23}-d_{12})/2).
\end{align}
\end{compress}
If we decompose the sum according to the parity of $d_{12}+d_{23}+d_{31}$, we get exactly sums of the form (\ref{eq:sumeven}) and (\ref{eq:sumodd}), up to boundary terms which vanish as $\ell\to 0$.
Therefore, using (\ref{eq:evenlim}) and (\ref{eq:oddlim}), we find that
\begin{equation}\label{eq:G3primitive}
\int_0^T\rmd T'\int_0^S\rmd S'\int_0^U\rmd U' \bar{G}_{g,3}^{(1)}(S',T',U') = F_g(S,T,U) - 1.
\end{equation}
Since $F_g(S,T,U)$ is a smooth function on $\R_+^3$ which takes value $1$ on the boundary, the three-point function $\bar{G}_{g,3}^{(1)}(S,T,U)$ is non-singular and can be obtained by differentiating (\ref{eq:G3primitive}), i.e.
\begin{equation}
\bar{G}_{g,3}^{(1)}(S,T,U) = \partial_S\partial_T\partial_U F_g(S,T,U). 
\end{equation}

As for the two-point function we will construct $\bar{G}_{g,3}^{(2)}(S,T,U)$ by considering the operation of removing the edges $e_1,e_2,e_3$ incident to the marked vertices of a weighted map $(\map,L)\in\metricmaps_{F,3}^{(1)}$.
Assuming $F\geq 2$, we distinguish two cases: either (A) no pair of the three edges shares a vertex, or (B) exactly one pair of edges shares a vertex.
It is not hard to see that the three-point function $\bar{G}_{g,3}^{(1)}(S,T,U)$ decomposes accordingly as
\begin{equation}\label{eq:G3decomp}
\bar{G}_{g,3}^{(1)}(S,T,U) = \bar{G}_{g,3}^{(A)}(S,T,U)+\bar{G}_{g,3}^{(B)}(S,T,U) + 2 g e^{-S-T-U},
\end{equation}
where $2 g e^{-S-T-U}$ is the contribution for $F=1$.
Suppose in case (B) that $e_1$ and $e_2$ share a vertex $v$, then $S=L(e_1)$, $T=L(e_2)$ and after removal of the edges $e_1$ and $e_2$ the remaining weighted map has two 1-valent vertices separated by a distance $U$.
Hence, we can express $\bar{G}_{g,3}^{(B)}(S,T,U)$ in terms of the two-point functions as
\begin{equation}\label{eq:G3B}
\bar{G}_{g,3}^{(B)}(S,T,U) = 2 e^{-T-U}G_{g,2}^{(1)}(S)+2 e^{-U-S}G_{g,2}^{(1)}(T)+2 e^{-S-T}G_{g,2}^{(1)}(U)-6 g e^{-S-T-U},
\end{equation}
where the last term is necessary to remove the contribution for $F=1$.
 
On the other hand, similarly to (\ref{eq:twopointid}), we have
\begin{equation}\label{eq:G3Aintegral}
\bar{G}_{g,3}^{(A)}(S,T,U) = 8 \int_0^{\infty}\!\!\rmd L_1e^{-L_1}\int_0^{\infty}\!\!\rmd L_2e^{-L_2}\int_0^{\infty}\!\!\rmd L_3e^{-L_3}\, \bar{G}_{g,3}^{(2)}(S-L_1,T-L_2,U-L_3).
\end{equation}
Combining (\ref{eq:G3decomp}) with (\ref{eq:G3B}) and (\ref{eq:G2-1zero2}) shows that $\bar{G}_{g,3}^{(A)}(S,0,0)=0$ and therefore (\ref{eq:G3Aintegral}) is solved by
\begin{compress}
\begin{align}
\bar{G}_{g,3}^{(2)}(S,T,U) =&\, \compactfrac{1}{8}(1+\partial_S)(1+\partial_T)(1+\partial_U)\bar{G}_{g,3}^{(A)}(S,T,U)\,+\, \compactfrac{1}{8}(1+\partial_T)(1+\partial_U)\bar{G}_{g,3}^{(A)}(0,T,U)\delta(S) \nonumber\\
&+ \compactfrac{1}{8}(1+\partial_S)(1+\partial_U)\bar{G}_{g,3}^{(A)}(S,0,U)\delta(T)\,+\, \compactfrac{1}{8}(1+\partial_S)(1+\partial_T)\bar{G}_{g,3}^{(A)}(S,T,0)\delta(U).
\end{align}
\end{compress}
In this expression we may replace $\bar{G}_{g,3}^{(A)}(\cdot,\cdot,\cdot)$ with $\bar{G}_{g,3}^{(1)}(\cdot,\cdot,\cdot)$, because the other terms in (\ref{eq:G3decomp}) are killed by the derivatives.
The desired expression (\ref{eq:G3-2expr}) then follows from setting $\bar{G}_{g,3}^{(1)}(S,T,0) = 2 G_{g,2}^{(1)}(S+T)$, which follows from $\partial_S\partial_T\partial_U F_g^{\mathrm{odd}}(S,T,U)|_{U=0}=0$ and 
\begin{equation}
\partial_S\partial_T\partial_U F_g^{\mathrm{even}}(S,T,U)|_{U=0} = \partial_S\partial_T\partial_U\left. \frac{\cfun^2(S+T+U)}{\cfun^2(S+T)}\right|_{U=0} = 2\partial_S^3 \log\cfun(S+T).
\end{equation}
\end{proof}

\begin{remark}
It is not hard to interpret the split of the 3-point functions into even and odd parts.
Consider a planar map $\map\in\maps_{F,3}$ with pair-wise distances $(d_{12},d_{23},d_{31})$ between the marked vertices.
Two scenarios are possible depending on whether a vertex $v\in\vertices(\map)$ exists such that $d_{ij}=d_{\map}(v_i,v)+d_{\map}(v_j,v)$.
If such a vertex exists, then $d_{12}+d_{23}+d_{31}$ is necessarily even.
Otherwise, in the limit in which $\map$ is very large, $d_{12}+d_{23}+d_{31}$ is even or odd with equal probability.
Translating to the weighted maps, where geodesics are almost surely unique, this means that weighted maps $(\map,L)\in\metricmaps_{F,3}^{(1)}$ contribute equally to $F^{\mathrm{even}}_g(S,T,U)$ and $F^{\mathrm{odd}}_g(S,T,U)$ except when the three geodesics meet in a point (the geodesics are ''completely confluent'').
In particular, the three-point function $\bar{G}_{g,3,\mathrm{conf}}^{(1)}(S,T,U)$ of weighted maps with completely confluent geodesics is given by
\begin{equation}\label{eq:G3conf}
\bar{G}_{g,3,\mathrm{conf}}^{(1)}(S,T,U)=\partial_S\partial_T\partial_U(F^{\mathrm{even}}_g(S,T,U) - F^{\mathrm{odd}}_g(S,T,U)).
\end{equation}
\end{remark}

Let us check some integrals of the three-point functions in Theorem \ref{thm:threepoint}. 
First of all
\begin{equation}
\int_0^{\infty} \rmd S\int_0^{\infty} \rmd T\int_0^{\infty} \rmd U \bar{G}_{g,3}^{(1)}(S,T,U) = F_g(S,T,U)\big|_{S=0}^{\infty}\big|_{T=0}^{\infty}\big|_{U=0}^{\infty} = \frac{g}{4\alpha \Sigma^{2}},
\end{equation}
which is precisely the generating function for $\nu_{F,3}(\metricmaps_{F,3}^{(1)})$ in (\ref{eq:measure1}).
The measure of the subset of $\metricmaps_{F,3}^{(2)}$ for which strict triangle inequalities hold is given by 
\begin{align}
\int_0^{\infty} \rmd S\int_0^{\infty} \rmd T\int_0^{\infty} \rmd U \compactfrac{1}{8}(1+\partial_S)(1+\partial_T)&(1+\partial_U)\bar{G}_{g,3}^{(1)}(S,T,U) \nonumber
\\&= g \frac{36 \alpha^2-4\alpha+24g-1}{64\alpha \Sigma^2}-\compactfrac{3}{4}(\alpha-\Sigma),
\end{align}
while for the subset with saturated triangle inequalities (i.e. the second line of (\ref{eq:G3-2expr})) it is
\begin{equation}
3 \int_0^{\infty}\rmd S\int_0^{\infty}\rmd T\,G_{g,2}^{(2)}(S+T) = g\frac{3\alpha-3}{4\alpha} +\compactfrac{3}{4}(\alpha-\Sigma).
\end{equation}
Adding these together we find
\begin{equation}
\int_0^{\infty} \rmd S\int_0^{\infty} \rmd T\int_0^{\infty} \rmd U \bar{G}_{g,3}^{(2)}(S,T,U) = g \frac{(1-2\alpha)(5+6\alpha-24\alpha^2)}{64\alpha \Sigma^2},
\end{equation}
which is the generating function for $\nu_{F,3}(\metricmaps_{F,3}^{(2)})$ in (\ref{eq:measure2}).

\subsection{Coincidence limit}\label{sec:coinc}

We already observed that $\bar{G}_{g,3}^{(1)}(S,T,0) = 2 G_{g,2}^{(1)}(S+T)$, which has the clear interpretation that the vertex $v_3$ must be located somewhere along the geodesic from $v_1$ to $v_2$.
A slightly tedious calculation, which we have included in the appendix, shows that
\begin{equation}
\partial_S\partial_U\bar{G}_{g,3}^{(1)}(S,T,U)\big|_{S=U=0} = \bar{G}_{g,3}^{(1)}(0,T,0) = 2 G_{g,2}^{(1)}(T).
\end{equation}
This in turn allows us to evaluate the limit $S,U\to 0$ (but $S,T,U\neq 0$) of $\bar{G}^{(2)}_{g,3}(S,T,U)$ as given in (\ref{eq:G3-2expr}), namely (for $T>0$)
\begin{align}
\lim_{\substack{S,U\to 0\\S,U\neq 0}}\bar{G}^{(2)}_{g,3}(S,T,U) &=  \compactfrac{1}{8}(1+\partial_S)(1+\partial_T)(1+\partial_U)\bar{G}_{g,3}^{(1)}(S,T,U)\big|_{S=U=0}\nonumber\\
&= \compactfrac{1}{8}(1+\partial_T)[(1+\partial_S)+(1+\partial_U)+\partial_S\partial_U-1]\bar{G}_{g,3}^{(1)}(S,T,U)\big|_{S=U=0}\nonumber\\
&=\compactfrac{1}{2}(1+\partial_T)^2G_{g,2}^{(1)}(T) = 2 G_{g,2}^{(2)}(T)\quad\quad(T>0).\label{eq:G3-2limit}
\end{align}
which, despite its simplicity, is actually not easily interpreted from the geometric point of view.
In particular the two-point function $G_{g,2}^{(2)}(T)$ does \emph{not} arise from configurations with one of the vertices lying on the geodesic connecting the other two, since $S,U\neq 0$ implies that we are enforcing strict triangle inequalities.
Instead we will see now that the important contributions are of the form of those in figure \ref{fig:threepointlimit}.

\begin{figure}[t]
\begin{center}
\includegraphics[width=.7\linewidth]{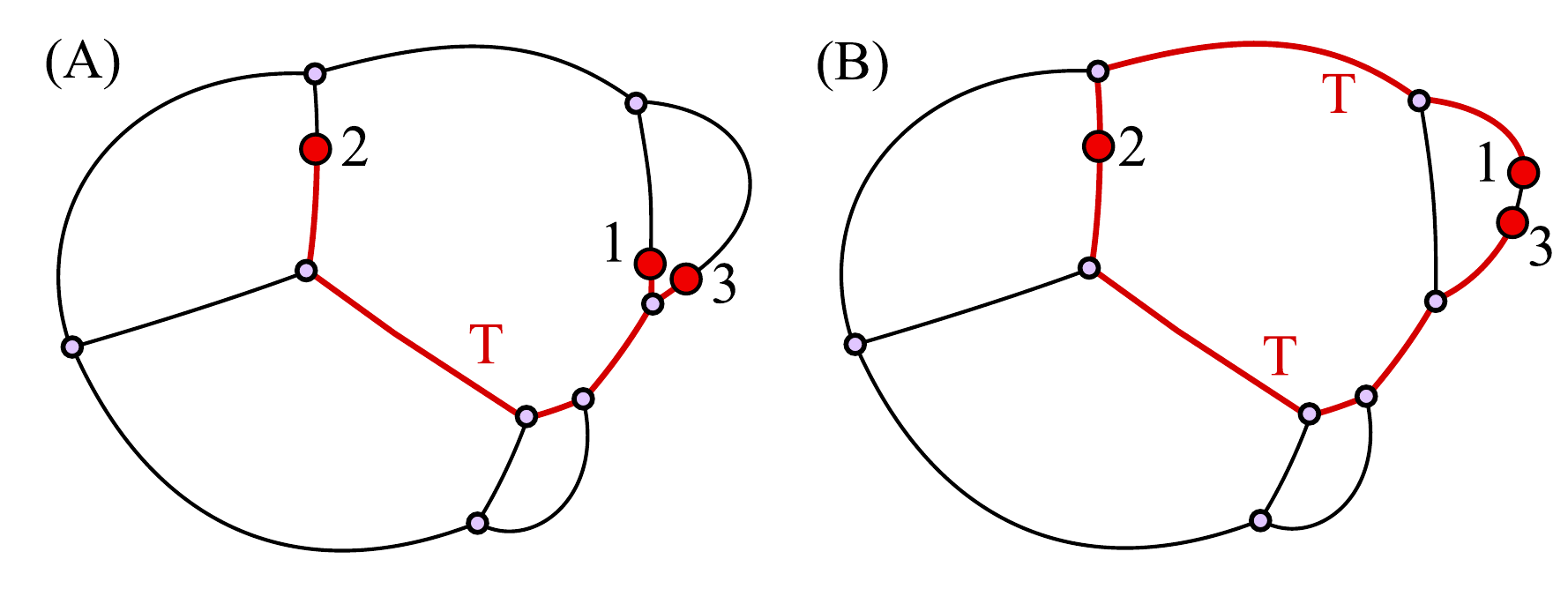}
\end{center}
\caption{Two possible ways the vertices $v_1$ and $v_3$ can be infinitesimally close without violating strict triangle inequalities in a weighted map $(\map,L,v_1,v_2,v_3)\in\metricmaps_{F,3}^{(2)}$.
The two geodesics of length $T$ connecting $v_1$ and $v_3$ to $v_2$ are colored red.}%
\label{fig:threepointlimit}
\end{figure}

Let us consider a weighted map $(\map,L)\in\metricmaps_{F,3}^{(2)}$ with geodesic distances corresponding to $S,T,U$, where both $S$ and $U$ are infinitesimally small but positive.
This corresponds to a situation where two of the marked vertices, say $v_1$ and $v_3$, are infinitesimally close, but strict triangle inequalities with the other vertex $v_2$ are maintained.
Two types of configurations of the three points are possible (as is illustrated in figure \ref{fig:threepointlimit}): (A) either $v_1$ and $v_3$ are adjacent to a cubic vertex $v$ and $d_{X_{\map,L}}(v_2,v)=T$, or (B) there is an (infinitesimally short) edge $e$ connecting $v_1$ and $v_3$ containing a local maximum of the distance function $d_{X_{\map,L}}(v_2,\cdot)$.
In case (A), merging $v_1$ and $v_3$ leads to an almost cubic weighted map with both a cubic and a bivalent marked vertex separated by a distance $T$.
Moreover, each such weighted map can be obtained in precisely two ways, since the same weighted map would have been obtained if $v_1$ and $v_3$ were interchanged.
In case (B), merging $v_1$ and $v_2$ leads to an almost cubic weighted map with two bivalent vertices separated by a distance $T$ that is conditioned to be realized by two geodesics.
In fact we may identify
\begin{align}
\lim_{S,U\to 0}\bar{G}^{(2)}_{g,3}(S,T,U) &= 2G_{g,2}^{(2,3)}(T) + 2G_{g,2}^{\text{max}}(T), \label{eq:G3-2limit2} \\
G_{F,2}^{\text{max}}(T) &:= \int_{\metricmaps_{F,1}^{(2)}} \rmd\nu_{F,1}(\map,L,v_1) \sum_{x\in\max(\map,L,v_1)} \delta(T-d_{X_{\map,L}}(v_1,x))
\end{align}
where $\max(\map,L,v_1)\subset X_{\map,L}$ is the finite set of local maxima of $d_{X_{\map,L}}(v_1,\cdot)$.

Recall that the two-point function $G_{F,2}^{(2)}(T)$ can be identified with the geometric two-point function $G^{\mathrm{geom}}_{F,2}(T)$ defined in (\ref{eq:geomtwopoint}) for $F\geq 3$.
In fact, we can represent $G_{F,2}^{(2)}(T)$ as 
\begin{equation}
G_{F,2}^{(2)}(T) = \int_{\metricmaps_{F,1}^{(2)}}\rmd\nu_{F,1}(\map,L,v_1)\,|\{x\in X_{\map,L} : d_{X_{\map,L}}(v_1,x) = T\}|.
\end{equation}
For a fixed weighted map the integrand only changes with increasing $T$ when either a vertex ($+1$) or a local maximum ($-2$) is encountered.
Therefore, for $T>0$,
\begin{equation}
\partial_TG_{F,2}^{(2)}(T) =\! \int_{\metricmaps_{F,1}^{(2)}}\!\!\!\rmd\nu_{F,1}(\map,L,v_1) \left(\sum_{v\in\vertices(\map)}\!\!\delta(T-d_{X_{\map,L}}(v_1,v))-2\!\!\!\!\!\!\!\!\!\!\!\sum_{x\in\max(\map,L,v_1)}\!\!\!\!\!\!\!\!\!\!\! \delta(T-d_{X_{\map,L}}(v_1,x)\right), 
\end{equation}
which implies the identity
\begin{equation}
\partial_TG_{g,2}^{(2)}(T) = G_{g,2}^{(2,3)}(T) - 2G_{g,2}^{\text{max}}(T).
\end{equation}
Combined with (\ref{eq:G3-2limit}) and (\ref{eq:G3-2limit2}) this leads to
\begin{align}
G_{g,2}^{(2,3)}(T) &= \compactfrac{1}{3}(2+\partial_T)G_{g,2}^{(2)}(T) = \compactfrac{1}{12}(2+\partial_T)(1+\partial_T)^2\partial_T^3\log\cfun(T), \label{eq:G2-23expr}\\
G_{g,2}^{\text{max}}(T) &= \compactfrac{1}{3}(1-\partial_T)G_{g,2}^{(2)}(T) = \compactfrac{1}{12}(1-\partial_T)(1+\partial_T)^2\partial_T^3\log\cfun(T). \label{eq:G2maxexpr}
\end{align}
We have still not managed to compute the two-point function for purely cubic weighted maps, but we are getting close and a clear pattern is emerging in the relation between the two-point functions.

\subsection{Collinear limit}\label{sec:coll}
\begin{figure}[t]
\begin{center}
\includegraphics[width=.35\linewidth]{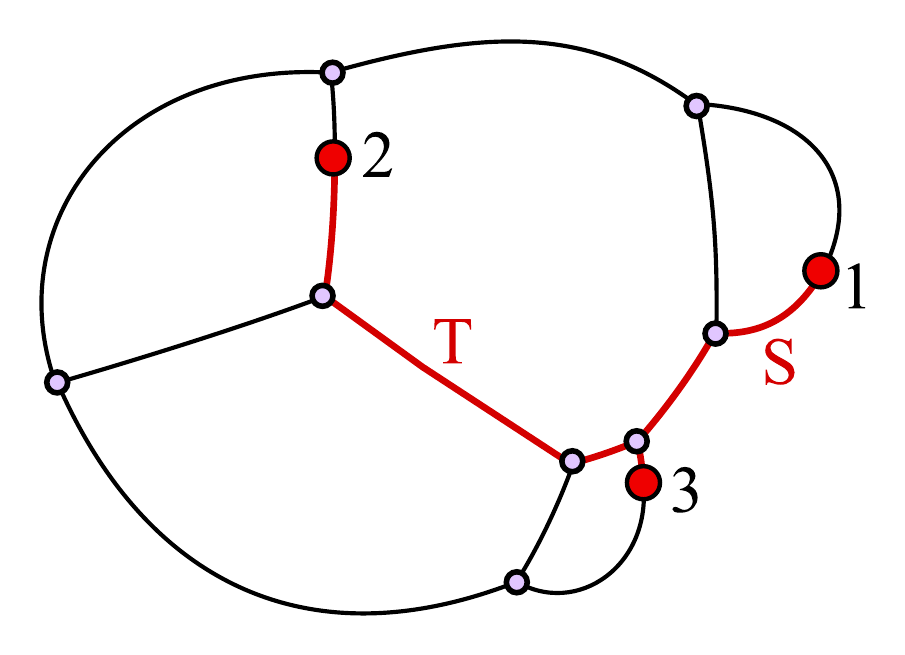}
\end{center}
\caption{If the geodesics connecting the three marked vertices of an almost cubic weighted map $(\map,L,v_1,v_2,v_3)\in\metricmaps_{F,3}^{(2)}$ are completely confluent and $U$ is infinitesimal (but nonzero), $v_3$ must be close to a cubic vertex on the geodesic from $v_1$ to $v_2$ but $v_3$ cannot be on that geodesic.}%
\label{fig:vertexongeodesic}
\end{figure}
Next let us consider just the limit $U\to 0$ of $\bar{G}_{g,3}^{(2)}(S,T,U)$ while $S,T>0$ are kept fixed.
This limit is of particular interest when we restrict to the completely confluent part of $\bar{G}_{g,3}^{(2)}(S,T,U)$, i.e.
\begin{equation}\label{eq:G3-2conf}
\bar{G}_{g,3,\mathrm{conf}}^{(2)}(S,T,U):=\compactfrac{1}{8}(1+\partial_S)(1+\partial_T)(1+\partial_U)\bar{G}_{g,3,\mathrm{conf}}^{(1)}(S,T,U),\quad\quad(S,T,U>0)
\end{equation}
with $\bar{G}_{g,3,\mathrm{conf}}^{(1)}(S,T,U)$ as in (\ref{eq:G3conf}).
In the limit $U\to 0$ the weighted maps $(\map,v_1,v_2,v_3)$ contributing to this three-point function must have their vertex $v_3$ infinitesimally close to a vertex $v$ that lies on a geodesic from $v_1$ to $v_2$ with distance $S$ to $v_1$ and $T$ to $v_2$ (see figure \ref{fig:vertexongeodesic}).
Hence, by integrating over $S$ while keeping $S+T$ fixed, one finds that 
\begin{equation}
G_{g,2,\mathrm{vert}}^{(2)}(T) := \lim_{U\to 0}\int_0^T\rmd S\,\bar{G}_{g,3,\mathrm{conf}}^{(2)}(S,T-S,U)
\end{equation}
gives the two-point function for weighted maps with two marked bivalent vertices connected by a geodesic of length $T$ and a marked cubic vertex on the geodesic.
The expected number of cubic vertices $\langle V \rangle_{F,T}$ on a geodesic of length $T$ in a random weighted map with $F$ faces is then given by
\begin{equation}
\langle V \rangle_{F,T} = \frac{G_{F,2,\mathrm{vert}}^{(2)}(T)}{G_{F,2}^{(2)}(T)}
\end{equation}
and, of course, the expected number of edges in the geodesic by $\langle V \rangle_{F,T}+1$.
Moreover, since the number of edges in the geodesic from $v_1$ to $v_2$ gives an upper bound for the graph distance $d_{\map}(v_1,v_2)$, the expected graph distance satisfies the inequality 
\begin{equation}\label{eq:graphdistbound}
\langle d_{\map}(v_1,v_2)\rangle_{F,T} < \langle V \rangle_{F,T}+1.
\end{equation}

Using (\ref{eq:G3-2conf}) and (\ref{eq:G3conf}) one finds
\begin{align}
\lim_{U\to 0}\bar{G}_{g,3,\mathrm{conf}}^{(2)}(S,T-S,U) &= G_{g,2}^{(2)}(T) + \compactfrac{1}{8}(1+\partial_S+\partial_T)(1+\partial_T)\partial_U\bar{G}_{g,3,\mathrm{conf}}^{(1)}(S,T-S,U)\big|_{U=0}
\end{align}
and therefore
\begin{align}
G_{g,2,\mathrm{vert}}^{(2)}(T) =&\, T G_{g,2}^{(2)}(T) - \compactfrac{1}{2}(1+\partial_T)\partial_TG_{g,2}^{(1)}(T) \nonumber\\
& +  \compactfrac{1}{8}(1+\partial_T)^2\int_0^T\rmd S\, \partial_U\bar{G}_{g,3,\mathrm{conf}}^{(1)}(S,T-S,U)\big|_{U=0}.
\end{align}
A tedious calculation shows that the latter integral is given by
\begin{align}
\int_0^T\rmd S\, \partial_U\bar{G}_{g,3,\mathrm{conf}}^{(1)}(S,T-S,U)\big|_{U=0} =&\, 4 \alpha T G_{g,2}^{(1)}(T) + 8 \log\left(\frac{\cfun(T)}{\Sigma}\right) \partial_T^3\log\left(\frac{\cfunp(T)}{\cfun(T)}\right) \nonumber\\
& - 4 \sinh(\Sigma T) \cfunp(T)\partial_T\left(\frac{\partial_T^2\log\cfunp(T)}{\cfun(T)^2}\right).
\end{align}
The explicit expression for $G_{g,2,\mathrm{vert}}^{(2)}(T)$ is perhaps of little interest, but it will allow us in Section \ref{sec:scaling} to obtain a simple expression for the scaling limit of the expected number of vertices on a geodesic.

\section{Relation to the Eden model on random triangulations}\label{sec:eden}

Given a planar map $\map$ we define an \emph{exploration process of length $k$} to be a sequence $\exploration=((\mathcal{V}_t,\mathcal{E}_t))_{t=0}^k$ of pairs consisting of subsets of \emph{explored vertices} $\mathcal{V}_t\subset\mathcal{V}(\map)$ and subsets of \emph{explored edges} $\mathcal{E}_t\subset\mathcal{E}(\map)$ satisfying the following properties:
\begin{enumerate}
\item Both endpoints of each edge $e\in\mathcal{E}_t$ are contained in $\mathcal{V}_t$.
\item For each $t=1,\ldots,k$, there exists an edge $e\notin \mathcal{E}_{t-1}$ such that at least one of its endpoints is in $\mathcal{V}_{t-1}$ and $\mathcal{E}_t = \mathcal{E}_{t-1}\cup\{e\}$.
\end{enumerate}
For each $t=0,\ldots,k$ we define the subset of \emph{frontier edges} $\vec{\mathcal{F}}_t\subset\vec{\edges}(\map)$ to contain the directed edges that start at vertices in $\mathcal{V}_t$, but of which the unoriented versions are not in $\mathcal{E}_t$ (see figure \ref{fig:explorationt}).
An exploration process is \emph{complete} if $\mathcal{E}_k = \mathcal{E}(\map)$.
An exploration process \emph{started at a vertex} $v$ is an exploration process with $\mathcal{V}_0 = \{v\}$ and $\mathcal{E}_0=\emptyset$.
Finally, by an exploration process of $\map$ \emph{started at an edge $e$} we mean an exploration process of $\map_1$ started at vertex $v$, where $\map_1$ is obtained from $\map$ by inserting a bivalent vertex $v$ in edge $e$ (see figure \ref{fig:explorationt} for an example).

\begin{figure}[t]
\begin{center}
\includegraphics[width=.9\linewidth]{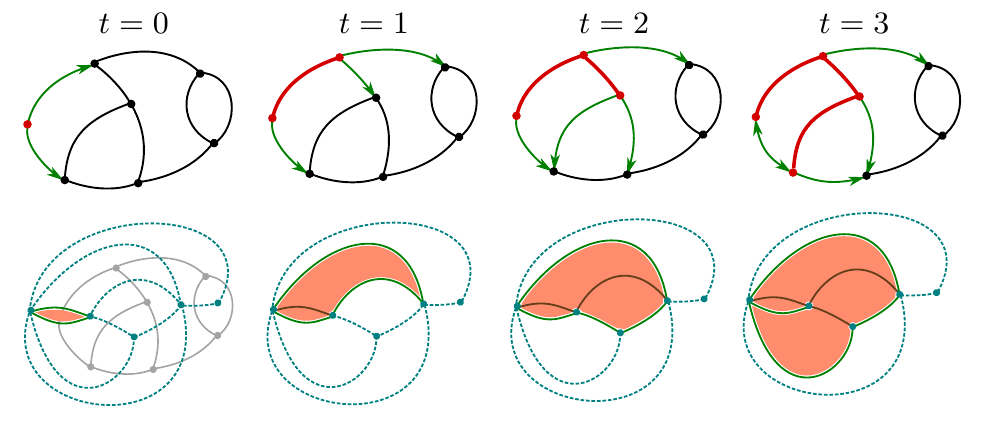}
\end{center}
\caption{The top row shows an example of an Eden model exploration process of length 3 started at an edge of a cubic planar map.
Explored edges and vertices are colored red and the green arrowed edges represent the oriented frontier edges.
The bottom row shows the same exploration process from the point of view of the dual triangulation. }%
\label{fig:explorationt}
\end{figure}

Clearly, an exploration process is fully determined by an initial pair $(\mathcal{V}_0,\mathcal{E}_0)$ and a (deterministic or probabilistic) algorithm that selects a frontier edge at each step $t$.
The \emph{Eden model exploration process} corresponds to the probabilistic algorithm that selects a frontier edge uniformly at random from $\vec{\mathcal{F}}_t$ at each step $t$.
One can either consider a complete Eden model process or stop the process according to some algorithm.
Of particular interest is the Eden model exploration process with \emph{stopping weight $\weight$}, where at each time $t$ the process is stopped with probability $\weight/(|\vec{\mathcal{F}}_t|+\weight)$.
Instead of just choosing a frontier edge uniformly at random, this comes down to randomly selecting either a frontier edge, each with weight $1$, or to stop with weight $\weight$.

It is well-known that the Eden model exploration process on a planar map $\map$ is closely related to the geodesic balls in $\map$ when the edge lengths are taken to be exponential random variables.
To be precise, let $\vertices_0\subset\vertices(\map)$ be a subset of the vertices of $\map$, and let $L:\edges(\map)\to\R_+$ be random edge lengths as in (\ref{eq:lengthmeasure}).
Define the \emph{max-distance} $d_{X_{\map,L}}(e,\vertices_0)$ of an edge $e\in\edges(\map)$ to the set $\vertices_0$ to be
\begin{equation}
d_{X_{\map,L}}(e,\vertices_0) = \max_{x\in e} \min_{v\in\vertices_0} d_{X_{\map,L}}(x,v),
\end{equation}
where by $x\in e$ we mean that $x\in X_{\map,L}$ is contained in the interval associated to the edge $e$.
Almost surely the values $d_{X_{\map,L}}(e,\vertices_0)$, $e\in\edges(\map)$, are distinct and therefore we can uniquely identify a complete exploration process $\exploration_L$ by setting $\edges_n$, $0\leq n\leq |\edges(\map)|$, equal to the set of $n$ edges which have smallest max-distance $d_{X_{\map,L}}(\cdot,\vertices_0)$.
Moreover, one obtains an exploration process $\exploration_{L,\weight}$ of random length $k$ by sampling in addition an exponentially distributed random variable $T$ with expectation value $1/\weight$, and stopping $\exploration_{L,\weight}$ after $k$ steps where $k$ is the number of edges with max-distance $d_{X_{\map,L}}(\cdot,\vertices_0)$ smaller than $T$.

\begin{lemma}\label{thm:edenmetric}
The random exploration process $\exploration_{L,\weight}$ described above is identical in law to the Eden model exploration process with stopping weight $\weight$ started at $\vertices_0$.
\end{lemma}
\begin{proof}
Let $k\geq 0$ be the length of the exploration process $\exploration_{L,\weight}$.
It suffices to check that the conditional probabilities
\begin{equation}\label{eq:probconditioned}
P(k\geq t\text{ and }\edges_t=\edges'\cup\{e\} \,|\, k \geq t-1\text{ and }\edges_{t-1}=\edges')
\end{equation}
agree for all $1\leq t\leq |\edges(\map)|$, and all possible $\edges'\subset\edges(\map)$ and edges $e$.
Notice that by construction this probability for the Eden model is 0, $1/(|\frontier_{t-1}|+\weight)$, or $2/(|\frontier_{t-1}|+\weight)$, depending on whether respectively none, one, or two of the orientations of $e$ are in the frontier $\frontier_{t-1}$.

Let us consider the exploration process $\exploration_{L,\weight}$ conditioned as in (\ref{eq:probconditioned}).
If we define
\begin{equation}
T_0 := \max_{e'\in\edges_{t-1}}d_{X_{\map,L}}(e',\vertices_0) \quad\text{and}\quad T_1 := \min_{e'\in \frontier_{t-1}}d_{X_{\map,L}}(e',\vertices_0),
\end{equation}
then the condition is equivalent to $T_0 \leq T$ and $T_0 < T_1$.
For each edge $e'\in\edges(\map)$ that has at least one of its orientations in $\frontier_{t-1}$ let us define $\Delta T(e')$ as follows.
If $e'\in\edges(\map)$ has precisely one of its orientations in $\frontier_{t-1}$ and $v$ is the endpoint of $e'$ for which $v\in\vertices_{t-1}$, let
\begin{equation}\label{eq:deltat1}
\Delta T(e'):=L(e') - T_0 + d_{X_{\map,L}}(v,\vertices_0).
\end{equation} 
On the other hand, if both orientations of $e'$ are in $\frontier_{t-1}$ and $v,v'\in\vertices_{t-1}$ are its endpoints, let 
\begin{equation}\label{eq:deltat2}
\Delta T(e'):=1/2(L(e') - 2T_0 + d_{X_{\map,L}}(v,\vertices_0)+ d_{X_{\map,L}}(v',\vertices_0)).
\end{equation}
Since $T_1 - T_0 = \min_{e'} \Delta T(e')$, the condition $T_0<T_1$ is equivalent to $\Delta T(e') >0 $ for all $e'$.
Since the $L(e')$ were originally independently, exponentially distributed with expectation value $1$, with this condition the $\Delta T(e')$ are exponentially distributed with expectation value $1$ in the case of (\ref{eq:deltat1}) and  expectation value $1/2$ in the case of (\ref{eq:deltat2}).
In particular, $T_1 - T_0$ is exponentially distributed with expectation value $1/|\frontier_{t-1}|$  and the situation $T_1 - T_0 = \Delta T(e')$ occurs with probability $1/|\frontier_{t-1}|$, respectively $2/|\frontier_{t-1}|$ for an edge $e'$ with one, respectively two, orientations in $e'\in\frontier_{t-1}$.

Moreover, since $T$ was exponentially distributed with expectation value $1/\weight$, conditioned on $T>T_0$ the probability that $T > T_1$ is equal to $|\frontier_{t-1}|/(|\frontier_{t-1}|+w)$.
Combining these probabilities we recover the probabilities mentioned above for the Eden model.
\end{proof}

\begin{lemma}
Given a planar map $\map$, two distinct vertices $v_1,v_2\in\vertices(\map)$ and $\weight\in\R_+$, the following three probabilities are equal, where $\exploration$ is an Eden model explorations process on $\map$:
\begin{enumerate}
\item[(a)] The probability that $\exploration$ reaches $v_2$, i.e. $v_2 \in \mathcal{V}_k$, when $\exploration$ is started at $v_1$ and has stopping weight $\weight$.
\item[(b)] The probability that the submap explored by $\exploration$ eventually becomes connected, when $\exploration$ is started at $\vertices_0=\{v_1,v_2\}$ and has stopping weight $2\weight$.
\item[(c)] The probability $P_{\map,\weight}(v_1,v_2)$ that $d_{X_{\map,L}}(v_1,v_2) < T$, when $T$ is an exponentially distributed random variable with expectation value $1/\weight$ and the edge lengths $L:\edges(\map)\to\R_+$ are independently and exponentially distributed with expectation value 1.
\end{enumerate}
\end{lemma}
\begin{proof}
Notice that the last probability $P_{\map,\weight}(v_1,v_2)$ is also the probability that an edge $e\in\edges(\map)$ exists that has $v_2$ as its endpoint and for which the max-distance $d_{X_{\map,L}}(e,v_1) < T$.
By Lemma \ref{thm:edenmetric} this is exactly the probability of (a). 
Second, observe that $d_{X_{\map,L}}(v_2,v_1)<T$ is equivalent to the condition that the set of edges $e$ for which $d_{X_{\map,L}}(e,\{v_1,v_2\})<T/2$ comprises a connected submap containing $v_1$ and $v_2$, but by Lemma \ref{thm:edenmetric} this is equivalent to situation (b).
\end{proof}
In the following we will simply refer to this probability as $P_{\map,w}(v_1,v_2)$ and we will use the convention $P_{\map,w}(v_1,v_2)=0$ when $v_1=v_2$.

Using these results we can interpret the two-point functions of weighted maps in terms of the Eden model exploration processes.
Indeed, the sum of the probabilities $P_{\map,\weight}(v_1,v_2)$ over all $(\map,v_1,v_2)\in\maps_{F,2}^{\bullet}$ is equal to the Laplace transform $\hat{G}_{F,2}^{\bullet}(w)$ of the corresponding two-point function $G_{F,2}^{\bullet}(T)$, i.e.
\begin{align}
\int_{\maps_{F,2}^{\bullet}}\!\!\rmd\nu_{F,2}(\map,v_1,v_2)P_{\map,\weight}(v_1,v_2) &= \int_{\metricmaps_{F,2}^{\bullet}}\!\rmd\nu_{F,2}(\map,L,v_1,v_2)e^{-\weight d_{X_{\map,L}}(v_1,v_2)}\nonumber\\
&=\int_0^{\infty}\!\!\!\rmd T\,e^{-\weight T}G_{F,2}^{\bullet}(T) =: \hat{G}_{F,2}^{\bullet}(w).\label{eq:G2laplfromP}
\end{align}
For instance, we can compute explicitly the Laplace transforms of $G_{g,2}^{(1)}(T)$ and $G_{g,2}^{(2)}(T)$ as follows.
We have
\begin{align}
\int_0^{\infty}\rmd T e^{-w T} \partial_T^2\log\cfun(T) & = \Sigma^2(\Sigma^2-\alpha^2) \int_0^\infty \rmd T \frac{e^{-w T}}{(\Sigma\cosh\Sigma T+\alpha \sinh\Sigma T)^2} \nonumber\\
&= -4\Sigma^2\beta \int_0^\infty \rmd T\frac{e^{-w T-2\Sigma T}}{(1-\beta e^{-2\Sigma T})^2} = -2\Sigma\beta \int_0^\beta\frac{\rmd x}{x} \frac{(x/\beta)^{\frac{w}{2\Sigma}+1}}{(1-x)^2}\nonumber\\
&= -2\Sigma \beta^{-\frac{w}{2\Sigma}} B_\beta\left(1+\frac{w}{2\Sigma},-1\right),
\end{align}
where we defined $\beta := (\alpha-\Sigma)/(\alpha+\Sigma)$ and $B_\beta\left(a,b\right)$ is the \emph{incomplete beta function}.
Therefore, using (\ref{eq:G2-1formula}) and (\ref{eq:G2-2formula}) and partial integrations to take care of the additional derivatives, we find
\begin{align}
&\hat{G}_{g,2}^{(1)}(w) = \alpha^2-\Sigma^2 -2 w\Sigma \beta^{-\frac{w}{2\Sigma}} B_\beta\left(1+\frac{w}{2\Sigma},-1\right), \\
&\hat{G}_{g,2}^{(2)}(w) = \compactfrac{1}{4}(1+w)^2\hat{G}_{g,2}^{(1)}(w) - \compactfrac{1}{4}(2+w+ \partial_T)G_{g,2}^{(2)}(T)|_{T=0}\nonumber\\
&\quad\quad=\frac{1-4\Sigma^2}{48}\left( (w-\alpha+1)^2+2\alpha^2-2\alpha + \compactfrac{1}{4}\right) - \frac{w(1+w)^2}{2}\Sigma  \beta^{-\frac{w}{2\Sigma}} B_\beta\left(1+\frac{w}{2\Sigma},-1\right).\label{eq:laplG2-2expr}
\end{align}

Given $(\map,v_1)\in\maps_{F,1}^{(2)}$ and $e\in\edges(\map)$, let $\map'\in\map_{F,2}^{(2)}$ be the weighted map obtained from $\map$ by inserting a bivalent vertex $v_2$ in the edge $e$. 
Then $P_{\map',w}(v_1,v_2)$ is equal to the probability $P(e\in\edges_k)$ that the edge $e$ is explored in an Eden model exploration process with stopping weight $w$ on $\map$ started at $v_1$.
This means that we can interpret $\hat{G}_{F,2}(w)$ as a sum over expectation values of the number of explored edges.
To be precise, we have the following result.

\begin{theorem}\label{thm:edenlength}
Let $\map$ be a uniformly random rooted cubic planar map with $F\geq 3$ faces. 
The expected length $\langle k \rangle_{F,\weight}$ of an Eden model exploration process with stopping weight $w$ started at the root edge is
\begin{equation}
\langle k \rangle_{F,\weight} = 2^{4-2F}\frac{F!(F-2)!!}{(3F-6)!!}\,\, [g^F]\hat{G}_{g,2}^{(2)}(w),
\end{equation}
where $[g^F]\hat{G}_{g,2}^{(2)}(w)$ is the coefficient of $g^{F}$ in (\ref{eq:laplG2-2expr}).
\end{theorem}
\begin{proof}
This is a simple combination of the remark above and (\ref{eq:measure2}) with $n=1$.
\end{proof}

For instance, the expectation values for $F=3$ and $F=4$ faces are
\begin{align}
\langle k \rangle_{3,\weight} &= -\frac{7}{w+2}+\frac{9}{2 (w+3)}+\frac{6}{(w+2)^2}+\frac{9}{2 (w+1)},\\
\langle k \rangle_{4,\weight} &= -\frac{37}{4 (w+2)}+\frac{27}{8 (w+3)}+\frac{9}{4 (w+4)}-\frac{21}{4 (w+2)^2}+\frac{81}{8 (w+3)^2}\nonumber\\
&\quad\quad+\frac{9}{(w+2)^3}+\frac{45}{8 (w+1)}+\frac{27}{8 (w+1)^2}.
\end{align}

\begin{figure}
\begin{center}
\includegraphics[width=\linewidth]{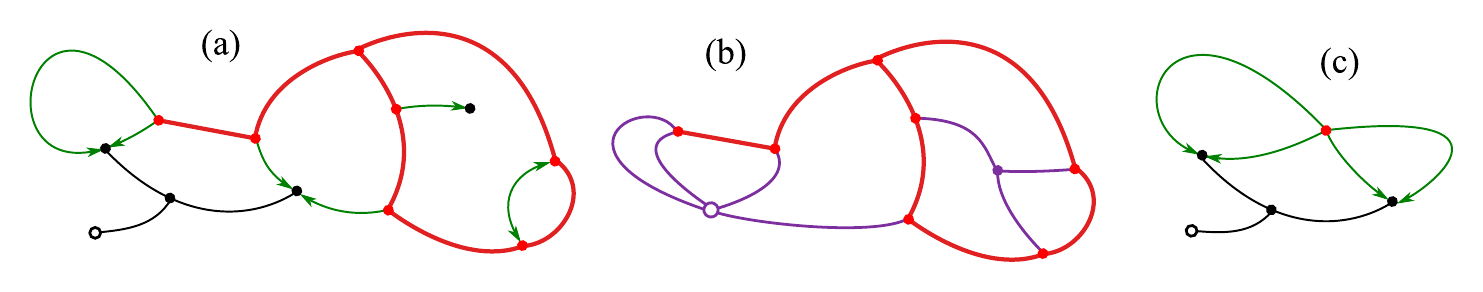}
\end{center}
\caption{(a) A possible state of an exploration process at time $t$ with an unexplored marked vertex represented by an open dot. (b) The explored submap $\map_t$ at time $t$ with two external vertices of which one is distinguished. (c) The unexplored submap represented by the distinguished external vertex. An exploration process starting at the red vertex of this planar map has the same probability of reaching the marked vertex as the one in (a) (given its current state).}
\label{fig:exploration}
\end{figure}
Let $\exploration=((\vertices_t,\edges_t))_{t=0}^k$ be an Eden model exploration process on a planar map $\map$.
We define the \emph{explored submap at time $t$} to be the planar map $\map_t$ obtained from the submap $\map_t'$ of $\map$ with edges $\edges_t$ by inserting a new marked vertex, called an \emph{external vertex}, in each face of $\map_t'$ that contains at least one frontier edge and a new edge starting at the starting corner of each frontier edge in $\vec{\mathcal{F}}_t$ and ending at one of the external vertices (see figure \ref{fig:exploration}b).
To each external vertex $v$ we can associate the \emph{unexplored submap represented by} $v$, which is given by contracting the edges of $\map$ that do not lie in the interior of the face of $\map'_t$ corresponding to $v$ to a single marked vertex (see Figure \ref{fig:exploration}c).
In case $\map$ possesses an unexplored marked vertex $v_1$, it is contained in a unique unexplored submap, which can then be regarded as having two marked vertices (the red dot and the open dot in Figure \ref{fig:exploration}c).

\begin{lemma}\label{thm:markovproperty}
Let $(\map,v_1,v_2)\in\maps_{F,2}^{(d_1,d_2)}$ be a random planar map (with respect to the measure $\nu_{F,2}$) and let $\exploration=((\vertices_t,\edges_t))_{t=0}^k$ be an Eden model exploration process started at $v_1$ with stopping weight $w$.
For $t\geq 0$, conditioned on $k\geq t$, on $v_2\notin\vertices_t$, on the explored submap $\map_t$ at time $t$, and on the unexplored submap containing $v_2$ being a member of $\maps_{F',2}^{(d_1',d_2)}$, the probability that $\exploration$ reaches $v_2$ equals the probability that a similar exploration process started at $v_1'$ on a random planar map $(\map',v_1',v_2')\in\maps_{F',2}^{(d_1',d_2)}$ w.r.t. the measure $\nu_{F',2}$ reaches $v_2'$.
\end{lemma}
\begin{proof}
Let $(\map,v_1,v_2)$ and $\exploration$ be conditioned as above.
By Lemma \ref{thm:edenmetric} the described probability $P$ is equal to the probability that $d_{X_{\map,L}}(\vertices_t,v_2)<T$ for random edge lengths $L$ and and an exponential random $T$ with mean $1/\weight$.
Since the shortest path from $v_2$ to $\vertices_t$ in $X_{\map,L}$ is necessarily contained in the face of $\map_t'$ corresponding to the external vertex $v$, its random length is identical in law to the distance between the marked vertices in the unexplored submap $(\map',v_1',v_2)$ if its edge lengths are chosen randomly with the same distribution.
Therefore $P$ is equal to the probability that an Eden model exploration process started at $v_1'$ on $\map'$ reaches $v_2$.
The result follows by noting that, for fixed $F'$ and $d_1'$, each $(\map',v_1',v_2')\in\maps_{F',2}^{(d_1',d_2)}$ can occur as unexplored submap and the probability distribution agrees with that of the measure $\nu_{F',2}$.
\end{proof}

For convenience let us introduce the notation $G_{F,2,\rt}^{(d_1,d_2)}(T)$ for the \emph{rooted} two-point function of almost cubic maps with one of the edges starting at $v_1$ marked as root edge, which is given by
\begin{equation}
G_{F,2,\rt}^{(d_1,d_2)}(T) := d_1\,G_{F,2}^{(d_1,d_2)}(T)
\end{equation}
and corresponding Laplace transform $\hat{G}_{F,2,\rt}^{(d_1,d_2)}(\weight)$.
Similarly, we define the \emph{rooted one-point function} as
\begin{equation}
W_{F}^{(d_1)} := d_1\,\nu_{F,1}(\maps_{F,1}^{(d_1)}),
\end{equation}
with the convention that $W_{F}^{(0)} = \delta_{F,1}$.
Notice that, with this convention, (\ref{eq:G2-generalzero}) implies that
\begin{equation}\label{eq:G2-generalzero2}
G_{F,2,\rt}^{(d_1,d_2)}(0) = d_1\, W_F^{(d_1+d_2-2)},\quad\quad(d_1,d_2\geq 1).
\end{equation}
Finally, let us introduce the generating functions
\begin{equation}\label{eq:genfunG2-general}
G_{g,\rt}(z_1,z_2;T) := \sum_{F=1}^{\infty} \sum_{d_1=1}^{\infty}\sum_{d_2=1}^{\infty} g^Fz_1^{-d_1-1}z_2^{-d_2-1} G_{F,2,\rt}^{(d_1,d_2)}(T)
\end{equation}
and
\begin{equation}\label{eq:genfunG1}
W_{g}(z_1) := \sum_{F=1}^{\infty} \sum_{d_1=0}^{\infty} g^Fz_1^{-d_1-1} W_{F}^{(d_1)}.
\end{equation}

\begin{proposition}\label{thm:loopequation}
For $d_1,d_2\geq 1$, the rooted two-point functions satisfy the equations
\begin{equation}\label{eq:loopequation1}
(d_1+\partial_T) G_{F,2,\rt}^{(d_1,d_2)}(T) = d_1\,G_{F,2,\rt}^{(d_1+1,d_2)}(T) + 2d_1 \sum_{F'=1}^{F-1}\sum_{d_1'=1}^{d_1-2}W_{F-F'}^{(d_1-d_1'-2)}\,G_{F',2,\rt}^{(d_1',d_2)}(T),
\end{equation}
and their generating functions satisfy
\begin{equation}\label{eq:loopequation2}
\frac{\partial}{\partial T} G_{g,\rt}(z_1,z_2;T) = \frac{\partial}{\partial z_1}\left[(z_1-z_1^2-2W_{g}(z_1))G_{g,\rt}(z_1,z_2;T)\right].
\end{equation}
\end{proposition}
\begin{proof}
Let us inspect the first step of an Eden model exploration process $\exploration$ with stopping weight $w$ started at the vertex $v_1$ of a random $(\map,v_1,v_2)\in\maps_{F,2}^{(d_1,d_2)}$.
With probability $\weight/(d_1+\weight)$ the process is killed immediately.
Otherwise $\edges_1 = \{e\}$ and $e$ is an edge starting at $v_1$ and ending at $v_1'$.
Three situations are now possible: (A) $v_1' = v_2$, (B) $v_1' \neq v_2$ and $v_1' \neq v_1$, or (C) $v_1'=v_1$.
The probability $P$ that $\exploration$ reaches $v_2$ decomposes accordingly as
\begin{equation}\label{eq:probdecompABC}
P := \frac{\hat{G}_{F,2,\rt}^{(d_1,d_2)}(\weight)}{\hat{G}_{F,2,\rt}^{(d_1,d_2)}(0)}=\frac{d_1}{d_1+\weight}\left(P_A+P_B+P_C\right).
\end{equation}
In case (A), $\exploration$ necessarily reaches $v_2$ therefore $P_A$ is simply the probability of case (A) occurring, which is
\begin{equation}
P_A=\frac{W_{F}^{(d_1+d_2-2)}}{\hat{G}_{F,2,\rt}^{(d_1,d_2)}(0)}=\frac{ G_{F,2,\rt}^{(d_1,d_2)}(0)}{d_1\,\hat{G}_{F,2,\rt}^{(d_1,d_2)}(0)},
\end{equation} 
where we used (\ref{eq:G2-generalzero2}).

In case (B), the explored submap at time $t=1$ contains one external vertex of degree $d_1+1$ representing an unexplored submap with $F$ faces.
Applying Lemma \ref{thm:markovproperty} we find
\begin{equation}
P_B = \frac{\hat{G}_{F,2,\rt}^{(d_1+1,d_2)}(\weight)}{\hat{G}_{F,2,\rt}^{(d_1,d_2)}(0)}.
\end{equation} 

In case (C), which can only occur when $d_1 \geq 3$ and $F\geq 2$, the explored submap at time $t=1$ contains two external vertices, $v$ and $v'$.
The sum of the degrees of the external vertices is $d_1-2$ and the total number of faces in both unexplored submaps is $F$.
Therefore, the probability that (C) occurs and that $v_2$ is in the unexplored submap represented by $v$ with $1\leq F'< F$ faces and $v$ having degree $1\leq d_1'\leq d_1-2$ is
\begin{equation}
\frac{W_{F-F'}^{(d_1-d_1'-2)}\,\hat{G}_{F',2,\rt}^{(d_1',d_2)}(0)}{\hat{G}_{F,2,\rt}^{(d_1,d_2)}(0)}.
\end{equation}
Hence, using Lemma \ref{thm:markovproperty},
\begin{equation}
P_C =2\sum_{F'=1}^{F-1}\sum_{d_1'=1}^{d_1-2} \frac{W_{F-F'}^{(d_1-d_1'-2)}\,\hat{G}_{F',2,\rt}^{(d_1',d_2)}(\weight)}{\hat{G}_{F,2,\rt}^{(d_1,d_2)}(0)},
\end{equation}
where the factor of 2 comes from the fact that $v_2$ can be in either of the two unexplored submaps.

Plugging the probabilities $P_A$, $P_B$, and $P_C$ into (\ref{eq:probdecompABC}) leads to
\begin{equation}
(d_1+\weight)\hat{G}_{F,2,\rt}^{(d_1,d_2)}(\weight)-G_{F,2,\rt}^{(d_1,d_2)}(0) = d_1\,\hat{G}_{F,2,\rt}^{(d_1+1,d_2)}(\weight)+2d_1\sum_{F'=1}^{F-1}\sum_{d_1'=1}^{d_1-2} W_{F-F'}^{(d_1-d_1'-2)}\,\hat{G}_{F',2,\rt}^{(d_1',d_2)}(\weight),
\end{equation}
but this is exactly the Laplace transform of equation (\ref{eq:loopequation1}).
By plugging (\ref{eq:genfunG2-general}) and (\ref{eq:genfunG1}) into (\ref{eq:loopequation2}) one can easily check that (\ref{eq:loopequation2}) is equivalent to (\ref{eq:loopequation1}).
\end{proof}

For $d_1\leq 2$, the sum in (\ref{eq:loopequation1}) vanishes and therefore Proposition \ref{thm:loopequation} implies
\begin{align}
G_{F,2}^{(2,d)}(T) &= \compactfrac{1}{2}(1+\partial_T)G_{F,2}^{(1,d)}(T),\\
G_{F,2}^{(3,d)}(T) &= \compactfrac{1}{3}(2+\partial_T)G_{F,2}^{(2,d)}(T).
\end{align}
Hence, we recover (\ref{eq:G2-23expr}) and moreover we obtain the following.
\begin{corollary}
The two-point function $G_{g,2}^{(3)}(T)$ for weighted cubic maps is given by
\begin{equation}
G_{g,2}^{(3)}(T) = \compactfrac{1}{36}(2+\partial_T)^2(1+\partial_T)^2\partial_T^3\log\cfun(T),
\end{equation}
with $\cfun(T)$ as in Theorem \ref{thm:twopoint12}.
\end{corollary}

\begin{remark}
In principle one can apply similar techniques to obtain the three-point function $\bar{G}_{g,3}^{(3)}(S,T,U)$ from $\bar{G}_{g,3}^{(2)}(S,T,U)$.
However, this requires a non-trivial investigation of edge-cases and an extension of Lemma \ref{thm:markovproperty} to disconnected explored submaps. 
We leave this to future investigation.
\end{remark}

\begin{remark}
Except for the statement in Lemma \ref{thm:markovproperty}, we have not utilized the general Markovian properties that the exploration process most likely possesses.
For instance, most of the results for the Eden model exploration process, like that of Theorem \ref{thm:edenlength}, should hold for any exploration process governed by a (deterministic or probabilistic) algorithm that selects a frontier edge independently of the unexplored part of the planar map.
See e.g. \cite{angel_growth_2002} for an investigation in the case of infinite planar triangulations. 
\end{remark}

\section{General two-point function}\label{sec:gentwop}

Equations very similar to (\ref{eq:loopequation2}) already appeared in \cite{watabiki_construction_1995,ambjom_scaling_1995} (see also \cite{ambjorn_quantum_1997} section 4.7.2) and were used to derive the continuum expression for the two-point function of triangulations for the first time.
In fact, (\ref{eq:loopequation1}) and (\ref{eq:loopequation2}) are identical to (4.336) and (4.337) in \cite{ambjorn_quantum_1997} after a renormalization of the parameters $g\to \sqrt{g}$ and $z_i\to z_i/\sqrt{g}$ to go from a factor of $g$ per vertex to a factor of $g$ per face.
It was shown in \cite{ambjom_scaling_1995} that (\ref{eq:loopequation2}) can be solved explicitly once $W_g(z)$ is known and an initial condition at $T=0$ is given.
It should be noted that the differential equations in \cite{ambjom_scaling_1995,ambjorn_quantum_1997} are written in terms of the \emph{two-loop function}, which in terms of our weighted cubic maps should be interpreted as an adapted two-point function with the condition on the contributing maps that the geodesic distance to the first marked point is decreasing along each edge leading away from the second marked point.
Both satisfy the same equation but with different initial conditions.
Therefore we can recycle the derivation in \cite{ambjom_scaling_1995,ambjorn_quantum_1997} using the initial condition (\ref{eq:G2-generalzero2}), which in terms of the generating functions reads
\begin{equation}\label{eq:twopointinitial}
G_{g,\rt}(z_1,z_2;0) = \frac{\partial}{\partial z_1}\left(\frac{1}{z_2}\frac{W_g(z_2)-W_g(z_1)}{z_2-z_1}\right).
\end{equation}

The rooted one-point function $W_g(z)$, which in the literature is often called the \emph{disk function} of random triangulations, is well-known and is, for instance, easily calculated using matrix models (see e.g. \cite{brezin_planar_1978}) or loop equations (see e.g. \cite{ambjorn_quantum_1997} section 4.2).  
In our notation it takes the form
\begin{equation}
W_g(z) = \frac{1}{2}\left(z-z^2+(z-\alpha-\compactfrac{1}{2})\sqrt{(z+\alpha-\compactfrac{1}{2})^2-2g/\alpha}\right),
\end{equation}
which can be checked to fall off like $g/z$ as $z\to\infty$.

The standard approach to solving the linear first-order partial differential equations like (\ref{eq:loopequation2}) is through the method of characteristics, i.e. we first solve the ordinary differential equation
\begin{equation}\label{eq:hatzdiff}
\hat{z}_1'(T)=\hat{z}_1(T)-\hat{z}_1(T)^2-2W_g(\hat{z}_1(T)), \quad\quad \hat{z}_1(0) = z_1.
\end{equation}
One can check that in the $z_1\to\infty$ limit the solution is given by
\begin{equation}\label{eq:hatzexpr}
\hat{z}(T) = \alpha+\compactfrac{1}{2} + \frac{\Sigma^2}{\sinh(\Sigma T)\cfun(T)}, 
\end{equation}
which is monotonically decreasing and approaches $\alpha +\compactfrac{1}{2}$ as $T\to\infty$.
The general solution for $z_1> \alpha+1/2$ is therefore given by
\begin{equation}
\hat{z}_1(T) = \hat{z}(T + \hat{z}^{-1}(z_1)).
\end{equation}
The unique solution to (\ref{eq:loopequation2}) with initial condition (\ref{eq:twopointinitial}) is then simply
\begin{equation}
G_{g,\rt}(z_1,z_2;T) = \frac{\partial}{\partial z_1}\left(\frac{1}{z_2}\frac{W_g(z_2)-W_g(\hat{z}_1(T))}{z_2-\hat{z}_1(T)}\right).
\end{equation}

We conclude that the generating function $G_{g}(z_1,z_2;T)$ for unrooted two-point functions is
\begin{align}
G_{g}(z_1,z_2;T) &:= \sum_{d_1=1}^{\infty}\sum_{d_2=1}^{\infty} z_1^{-d_1-1}z_2^{-d_2-1} G_{g,2}^{(d_1,d_2)}(T)\nonumber\\
& = \frac{-1}{z_1z_2}\left(\frac{W_g(z_2)-W_g(\hat{z}_1(T))}{z_2-\hat{z}_1(T)}-\frac{W_g(z_2)-W_g(\hat{z}(T))}{z_2-\hat{z}(T)}\right),
\end{align} 
where the second term is needed to ensure $\lim_{z_1\to\infty}z_1G_{g}(z_1,z_2;T)=0$.
One can check explicitly that this expression is symmetric in $z_1$ and $z_2$ and that it reproduces the previously derived two-point functions.
Let us illustrate the case where the second marked vertex is univalent, corresponding to the coefficient of $z_2^{-2}$ in $G_{g}(z_1,z_2;T)$.
Combining (\ref{eq:hatzdiff}) and (\ref{eq:hatzexpr}) one may deduce that
\begin{equation}
\partial_TW_g(\hat{z}(T)) = \compactfrac{1}{2}\partial_T( \hat{z}'(T)-\hat{z}(T)+\hat{z}(T)^2 ) = \partial_T^3\log\cfun(T) = G_g^{(1)}(T).
\end{equation}
Using the explicit expansion
\begin{equation}
\hat{z}^{-1}(z_1) = \frac{1}{z_1} + \frac{1}{2z_1^2} + \frac{1}{3z_1^3} + \mathcal{O}(z_1^{-4}),
\end{equation}
one finds that 
\begin{align}
&[z_2^{-2}]G_{g}(z_1,z_2;T) = \frac{1}{z_1}[ W_g(\hat{z}_1(T))-W_g(\hat{z}(T)) ] \nonumber\\
&\quad\quad= \frac{1}{z_1}\left[ W_g(\hat{z}(T+\compactfrac{1}{z_1}+\compactfrac{1}{2z_1^2}+\compactfrac{1}{3z_1^3}+\ldots))-W_g(\hat{z}(T))\right]\nonumber\\
&\quad\quad= G_g^{(1)}(T) z_1^{-2} + \compactfrac{1}{2}(1+\partial_T)G_g^{(1)}(T) z_1^{-3} + \compactfrac{1}{6}(2+\partial_T)(1+\partial_T)G_g^{(1)}(T) z_1^{-4} + \mathcal{O}(z_1^{-5}),
\end{align}
in agreement with the formulas in Section \ref{sec:twop}.

\section{Large graph limits}\label{sec:scaling}

All generating functions that we have encountered have a radius of convergence equal to $g^{*}=\frac{1}{12\sqrt{3}}$ and are non-analytic at $g=g^{*}$.
As usual one can study the asymptotic growth of their coefficients by expanding around $g=g^{*}$.
Writing $g = g^{*}(1-24\epsilon^2)$ and expanding $\alpha$ and $\Sigma$ around $\epsilon=0$ yields
\begin{align}
\alpha &= \frac{1}{2\sqrt{3}} + \frac{2\epsilon}{\sqrt{3}} - \frac{4\epsilon^2}{3\sqrt{3}}+\mathcal{O}(\epsilon^3), \\
\Sigma &= \epsilon^{\frac{1}{2}} + \compactfrac{2}{3}\epsilon^{\frac{3}{2}} - \epsilon^{\frac{5}{2}}+\mathcal{O}(\epsilon^{\frac{7}{2}}).
\end{align}

\subsection{scaling limit}
If in addition we scale distances like $T = T_0 \epsilon^{-\frac{1}{2}}$ we find that
\begin{equation}
\cfun(T) = \frac{1}{\sqrt{12}} \sinh T_0 + 2 \sqrt{\epsilon}\cosh T_0 + \mathcal{O}(\epsilon)
\end{equation} 
and the various two-point functions are two leading order given by
\begin{align}
G_{g,2}^{(1)}(T) &= 2\, \frac{\cosh T_0}{\sinh^3 T_0}\, \epsilon^{\frac{3}{2}} + \mathcal{O}(\epsilon^2),\label{eq:conttwopointfunction}\\
G_{g,2}^{(2)}(T) &= \frac{1}{2}\, \frac{\cosh T_0}{\sinh^3 T_0}\, \epsilon^{\frac{3}{2}} + \mathcal{O}(\epsilon^2),\\
G_{g,2}^{(3)}(T) &= \frac{2}{9}\, \frac{\cosh T_0}{\sinh^3 T_0}\, \epsilon^{\frac{3}{2}} + \mathcal{O}(\epsilon^2).
\end{align}
Up to numerical factors these agree exactly with the scaling of the two-point functions of triangulations \cite{ambjom_scaling_1995}, quadrangulations \cite{bouttier_geodesic_2003} and general planar maps \cite{ambjorn_trees_2013,bouttier_two-point_2013,fusy_three-point_2014}.
Hence, in the large $F$ limit, the distribution of distances between the marked vertices agrees with that of the Brownian map.

On the other hand $F^{\mathrm{even}}_g(S,T,U)$ and $F^{\mathrm{odd}}_g(S,T,U)$ are identical to leading order,
\begin{equation}\label{eq:Flimit}
F^{\mathrm{even/odd}}_g(S,T,U) = \frac{1}{12}\,\frac{\sinh^2S_0\sinh^2T_0\sinh^2U_0\sinh^2(S_0+T_0+U_0)}{\sinh^2(S_0+T_0)\sinh^2(T_0+U_0)\sinh^2(U_0+S_0)} \epsilon^{-1} + \mathcal{O}(\epsilon^{-\frac{1}{2}}),
\end{equation}
while the leading orders of $\bar{G}_{g,3}^{(1)}(S,T,U)$ and $\bar{G}_{g,3}^{(2)}(S,T,U)$ are obtained by acting on (\ref{eq:Flimit}) with $\epsilon^{\frac{3}{2}}\partial_{S_0}\partial_{T_0}\partial_{U_0}$ and $\compactfrac{1}{8}\epsilon^{\frac{3}{2}}\partial_{S_0}\partial_{T_0}\partial_{U_0}$ respectively.
Again we observe that in the scaling limit, up to numerical factors, the two- and three-point functions agree exactly with those of quadrangulations \cite{bouttier_three-point_2008} and general planar maps \cite{fusy_three-point_2014}.

One can check that
\begin{align}
\lim_{U\to 0}\bar{G}^{(2)}_{g,3,\mathrm{conf}}(S,T-S,U) &= \frac{1}{2}\left(1+\frac{1}{\sqrt{3}}\right) \frac{\cosh T_0}{\sinh^3 T_0}\,\epsilon^{\frac{3}{2}} + \mathcal{O}(\epsilon^2)\nonumber\\
&= \left(1+\frac{1}{\sqrt{3}}\right) G_{g,2}^{(2)}(T) + \mathcal{O}(\epsilon^2).\label{eq:vertexdensity}
\end{align}
Hence
\begin{equation}\label{eq:g2vertasym}
G^{(2)}_{g,2,\mathrm{vert}}(T) = \left(1+\frac{1}{\sqrt{3}}\right) T G_{g,3}^{(2)}(T) + \mathcal{O}(\epsilon^{\frac{3}{2}}).
\end{equation}
We will discuss the consequences of this simple relation in (\ref{sec:discussion}).

\subsection{Local limit}
Finally, let us have a look at the $F\to\infty$ limit of the two-point function $G_{F,2}^{(1)}(T)$ while keeping $T$ fixed, also known as the \emph{local limit} of $G_{F,2}^{(1)}(T)$.
It can again be found by scaling $g = g^{*}(1-24\epsilon^2)$ in $G_{g,2}^{(1)}(T)$ and determining the first non-analytic term in the expansion in $\epsilon$, which turns out to be the one proportional to $\epsilon^3$.
Performing this calculation shows that to leading order in $F$,
\begin{equation}
G_{F,2}^{(1)}(T) \sim \frac{F^{-\frac{5}{2}} (12\sqrt{3})^F}{630\sqrt{2\pi}} \left(3+23 T'+30 T'^2+10 T'^3 - \frac{3}{(1+T')^4}\right),\quad\quad T' = \frac{T}{2\sqrt{3}},
\end{equation}
while analogous formulas for $G_{F,2}^{(2)}(T)$ and $G_{F,3}^{(3)}(T)$ are obtained by taking derivatives with respect to $T$.

As an application one may obtain the $F\to\infty$ limit of the expected length $\langle k\rangle_{F,w}$ of an Eden model exploration process with stopping weight $w$ started at the root edge of a random cubic planar map with $F$ faces.
Using Theorem \ref{thm:edenlength} we find
\begin{equation}
\lim_{F\to\infty}\langle k\rangle_{F,w} = \frac{p_0(w)}{w^4} - \frac{144\sqrt{3}}{35} w^3(1+w)^2 e^{-2\sqrt{3}w}\Ei(2\sqrt{3}w),
\end{equation}
where $\Ei(\cdot)$ is the exponential integral function and $p_0(w)$ is the eighth-order polynomial
\begin{align}
p_0(w)&=\frac{1}{35}\Big(5 + 10w + 28w^2 + 46w^3 - 24w^4 - 24w^5 - 84w^6 - 144w^7 - 72w^8\nonumber\\
&\quad\quad\quad +2\sqrt{3}\, w(1+w)^2(5+3w^2+6w^4)\Big).
\end{align}

\section{Discussion}\label{sec:discussion}

We have derived explicit expressions for the two- and three-point functions for random (almost) cubic maps with exponential edge weights.
In the grand-canonical scaling limit they agree with the analogous expressions for the Brownian map.
This suggest that the deformation of the geometry introduced by the random edge weights has no effect, other an overall change of scale, on the continuum geometry.

Apart from the novel enumeration techniques for cubic maps, the current study also sheds some light on the relation between inequivalent distances on the same cubic map.
To see this let us return to the scaling limit (\ref{eq:g2vertasym}) of the two-point function $G_{g,2,\mathrm{vert}}^{(2)}(T)$ with a marked vertex on the geodesic of length $T$.
Although we do not give a proof, one may deduce from (\ref{eq:g2vertasym}) that the same linear relation holds between the large-$F$ limits of $G^{(2)}_{F,2,\mathrm{vert}}(T)$ and $TG^{(2)}_{F,2}(T)$ as long as one keeps $TF^{-1/4}$ constant.
Therefore the expected number $\langle V \rangle_{F,T}$ of cubic vertices on a geodesic of length $T$ scales linearly with $T$ with ratio $1+1/\sqrt{3}$, i.e.
\begin{equation}
\lim_{\substack{F\to\infty\\ TF^{-1/4}\text{ const}}} \frac{\langle V \rangle_{F,T}}{T} = 1+\frac{1}{\sqrt{3}}.
\end{equation}
Of course, this formula should also hold for the situation where the marked vertices are cubic.
It then follows from (\ref{eq:graphdistbound}) that the expected graph distance $\langle d_{\map}(v_1,v_2) \rangle_{F,T}$ is bounded from above by $(1+1/\sqrt{3})T$ as $F\to\infty$.
On the other hand, by the central limit theorem, conditioned on $d_{\map}(v_1,v_2)$ almost surely $T\leq d_{\map}(v_1,v_2)$ as $F\to\infty$.
Therefore, $\langle d_{\map}(v_1,v_2) \rangle_{F,T}$ is asymptotically bounded from below by $T$.
We conclude that 
\begin{equation}
\liminf_{\substack{F\to\infty\\ TF^{-1/4}\text{ const}}}\frac{\langle d_{\map}(v_1,v_2) \rangle_{F,T}}{T}\geq 1 \quad\text{and}\quad  \limsup_{\substack{F\to\infty\\ TF^{-1/4}\text{ const}}}\frac{\langle d_{\map}(v_1,v_2) \rangle_{F,T}}{T}\leq 1+\frac{1}{\sqrt{3}}.
\end{equation} 

\begin{figure}
\begin{center}
\includegraphics[width=0.5\linewidth]{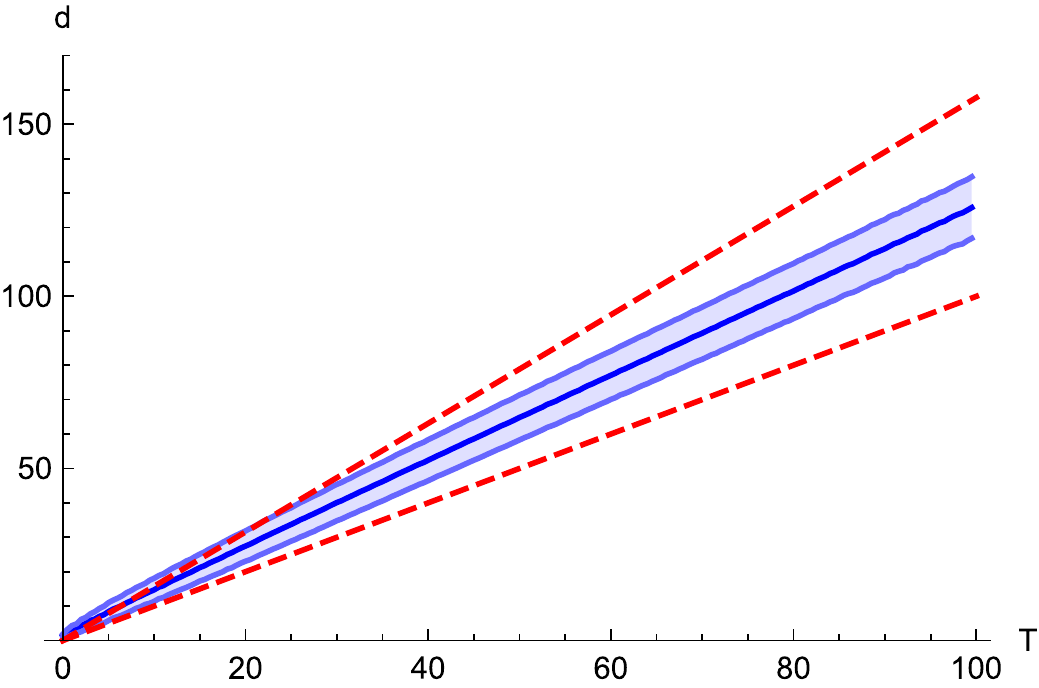}
\end{center}
\caption{The numerical expectation value (dark curve) together with the standard deviation (shaded area) of the graph distance $d$ between random pairs of vertices that are conditioned to have geodesic distance (in a small interval around) $T$ in a random weighted cubic map with $F=32000$ faces. The dashed lines correspond to $d=T$ and $d=(1+1/\sqrt{3})T$ respectively.}
\label{fig:simulation}
\end{figure}

Notice that (\ref{eq:vertexdensity}) does not only tell us about the expect total number of vertices on the geodesic, it also implies that the expected \emph{density} of vertices is equal to $1+1/\sqrt{3}$ as $F\to\infty$ independently of the position $S$ along the geodesic. 
It is quite plausible that the densities at different distances are independent enough to make the average density of vertices along the geodesic almost surely equal to $1+1/\sqrt{3}$.
If that is indeed the case, one may conclude that in the limit $F\to\infty$ the ratio of the graph distance and geodesic distance between a random pair of vertices is almost surely between $1$ and $1+1/\sqrt{3}$. 

This suspicion can be easily checked numerically using a Monte Carlo simulation, as can be seen in Figure \ref{fig:simulation} for a random weighted cubic map with $F=32000$ faces.
In fact, the numerics suggest that the graph distance is precisely in the middle of the two bounds, i.e. that ratio is $1+1/(2\sqrt{3})$, with high probability.
Without going into details let us mention that the ratio $1+1/(2\sqrt{3})$ can in fact be obtained heuristically by comparing the continuum limit of the differential equation (\ref{eq:loopequation2}) to the analogous equation (19) in \cite{kawai_transfer_1993}, related to a transfer matrix for the graph distance on cubic maps.

In addition to the graph distance $d_\map$ and geodesic distance (or first-passage time) $T$ on the cubic map, there exist another natural distance $d_\map^*$, which is the graph distance on the triangulation dual to the cubic map.
Based on a comparison of the continuum two-point function (\ref{eq:conttwopointfunction}) to that of large random triangulations, which can e.g. be deduced from \cite{gall_uniqueness_2013}, Theorem 1.1, we expect the ratio of the geodesic distance $T$ to the triangulation-distance to be given by $2\sqrt{3}$.

It would be interesting to see whether there is any universality to the simple ratios between the various distances. 
In a forthcoming paper we will study these relations in the more general setting of Boltzmann planar maps, where vertices of arbitrary degree are allowed.

\subsection*{Acknowledgments}
We are grateful to an anonymous referee for many useful suggestions.
The authors acknowledge support from the ERC-Advance grant 291092,
``Exploring the Quantum Universe'' (EQU). JA acknowledges support
of FNU, the Free Danish Research Council, from the grant
``quantum gravity and the role of black holes''.   
In addition JA was supported 
in part by Perimeter Institute of Theoretical Physics.
Research at Perimeter Institute is supported by the Government of Canada
through Industry Canada and by the Province of Ontario through the 
Ministry of Economic Development \& Innovation.
 
\appendix
\section{Coincidence limit calculation}
In this appendix we will show the non-trivial fact that $\partial_S\partial_U\bar{G}_{g,3}^{(1)}(S,T,U)\big|_{S=U=0} = \bar{G}_{g,3}^{(1)}(0,T,0)$.
The odd part is easy since the derivatives $\partial_S^2\partial_U^2$ need to act on the factor $\sinh^2\Sigma S\sinh^2\Sigma U$ in order to make the result nonzero.
Hence
\begin{align}
\partial_T\partial_S^2\partial_U^2F_g^{\mathrm{odd}}(S,T,U)|_{S=U=0} &= 4\partial_T\left[(\alpha^2-\Sigma^2)^2\sinh^2\Sigma T \frac{\cfunp^2(T)}{\cfun^4(T)}\right]\\
&= 4\partial_T\left[ (\cfun'(T)-\alpha \cfun(T))^2 \frac{(\cfun'(T)+\alpha \cfun(T))^2}{\cfun^4(T)}\right]\\
&= 4\partial_T\left[ \left(\frac{\cfun'(T)}{\cfun(T)}\right)^2-\alpha^2\right]^2, \label{eq:fevencoin}
\end{align}
where in the second equality we used that $\cfunp(T) = \cfun'(T)+\alpha \cfun(T)$.
The even part requires more work:
\begin{align}
&\partial_T\partial_S^2\partial_U^2F_g^{\mathrm{even}}(S,T,U)|_{S=U=0}=\partial_T\partial_S^2\Bigg[ \partial_U^2 \frac{\cfun^2(S+T+U)}{\cfun^2(S+T)}+2\left(\partial_U\frac{\cfun^2(T)}{\cfun^2(T+U)}\right)\partial_U\frac{\cfun^2(S)}{\cfun^2(S+U)}\nonumber\\
&\quad+ 2\left(\partial_U\frac{\cfun^2(S+T+U)}{\cfun^2(S+T)}\right)\partial_U\left(\frac{\cfun^2(U)}{\Sigma^2}+\frac{\cfun^2(T)}{\cfun^2(T+U)}+\frac{\cfun^2(S)}{\cfun^2(S+U)}\right)\Bigg]_{U=S=0}\\
&=\partial_T\partial_S^2\Bigg[  \frac{\partial_T^2\cfun^2(T+S)}{\cfun^2(T+S)}+8\frac{\cfun'(T)}{\cfun(T)}\frac{\cfun'(S)}{\cfun(S)}+8\frac{\cfun'(S+T)}{\cfun(S+T)} \left(\frac{\cfun'(0)}{\cfun(0)}-\frac{\cfun'(T)}{\cfun(T)}-\frac{\cfun'(S)}{\cfun(S)}\right)\Bigg]_{S=0}\\
&=\partial_T\left[ \partial_T^2 \frac{\partial_T^2\cfun^2(T)}{\cfun^2(T)}-8\left(\partial_T^2\frac{\cfun'(T)}{\cfun(T)}\right)\frac{\cfun'(T)}{\cfun(T)} - 16 \left(\partial_T\frac{\cfun'(T)}{\cfun(T)}\right)\partial_S\frac{\cfun'(S)}{\cfun(S)}\right]_{S=0}\nonumber\\
&= 4\partial_T\left[ \left(\partial_T \frac{\cfun'(T)}{\cfun(T)}\right)^2-\frac{\cfun'(T)}{\cfun(T)}\partial_T^2\frac{\cfun'(T)}{\cfun(T)}-4(\Sigma^2-\alpha^2)\partial_T\frac{\cfun'(T)}{\cfun(T)}\right]\\
&= 4\partial_T\left[ -\left(\frac{\cfun'(T)}{\cfun(T)}\right)^4 +4(\Sigma^2-\alpha^2)\left(\frac{\cfun'(T)}{\cfun(T)}\right)^2\right],\label{eq:foddcoin}
\end{align}
where we used $\cfun''(T) = \Sigma^2 \cfun(T)$ and $\partial_S \cfun'(S)/\cfun(S)|_{S=0} = \Sigma^2-\alpha^2$.
Combining (\ref{eq:fevencoin}) and (\ref{eq:foddcoin}) and using that $2\Sigma^2 = 3\alpha^2-1/4$ we obtain
\begin{align}
\partial_T\partial_S^2\partial_U^2F_g(S,T,U)|_{S=U=0} = 4\partial_T\left[ (4\Sigma^2-6\alpha^2)\left( \frac{\cfun'(T)}{\cfun(T)}\right)^2\right] = 2 \partial_T^3\log \cfun(T).
\end{align}
 
\bibliographystyle{hsiam}
\bibliography{eden}

\end{document}